\documentclass[runningheads]{llncs}
\usepackage[T1]{fontenc}
\usepackage{graphicx}
\usepackage{booktabs}
\usepackage[misc]{ifsym}
\usepackage{tikz}
\usepackage[noend,ruled,linesnumbered]{algorithm2e}
\usepackage{multirow}
\usepackage{subcaption}
\usepackage{graphicx}
\usepackage{pifont}
\usepackage{amsmath}
\usepackage{amssymb}
\usepackage{bbding}
\newcommand{\Sref}[1]{\S\ref{#1}}

\usepackage{mwe}

\newcommand\blfootnote[1]{
  \begingroup
  \renewcommand\thefootnote{}
  \hypersetup{hyperfootnotes=false}
  \footnote{#1}
  \endgroup
}

\usepackage{hyperref}
\usepackage{booktabs}
\usepackage{pifont}
\usepackage{makecell}
\begin{document}

\title{Protecting K-Nearest Neighbor Queries from Location Inference Attacks}

\author{Zhiyu Sun\inst{1*}  \and Jie Fu\inst{2*} \and Xinpeng Ling\inst{3} \and Huifa Li\inst{4} \and Zhili Chen\inst{1} (\Letter)}

\institute{East China Normal University, China \\
\email{zhlchen@sei.ecnu.edu.cn}
\and
Stevens Institute of Technology, USA 
\and
Tongji University, China 
\and
Mohamed bin Zayed University of Artificial Intelligence, UAE}

\maketitle
\blfootnote{$^*$ Both authors contributed equally to this work.} 

\begin{abstract}
The k-nearest neighbor query (kNNQ) is a core component of modern location-based services (LBS) and has been widely adopted in popular features such as ``people nearby''. However, its potential privacy risks have long been overlooked. In this work, we present the first two attacks against kNNQ, namely the geometric intersection location inference attack (GI-LIA) and the zero-order optimization location inference attack (ZO-LIA), revealing the inherent location privacy risks posed by kNNQ. To mitigate these privacy risks, we further propose DPRS, a differential privacy framework for kNNQ protection. The core idea of DPRS is to incorporate a rejection sampling mechanism within a constrained perturbation interval, thereby mitigating the distance distortion caused by excessive noise injection. In addition, we design a private interval construction algorithm to construct the perturbation interval, enabling the rejection sampling mechanism to achieve a more favorable trade-off between privacy protection and query utility in kNNQ. Extensive experiments on real-world spatial datasets demonstrate that DPRS outperforms existing methods in both privacy protection and query utility. Our code is available at \url{https://github.com/reanatom/DPRS}.

\keywords{ K-nearest neighbor query \and Location inference attack  \and Differential privacy }
\end{abstract}

\vspace{-2mm}
\section{Introduction}\label{sec:prelim}

Location-Based Services (LBS) have become a fundamental component of modern mobile applications~\cite{shi2021efficient,9101624,10.1007/s00778-019-00568-7}, but they also introduce significant privacy risks. Prior work has demonstrated that attackers can infer users' true locations from obfuscated location information. Existing location inference attacks (LIAs) typically exploit distance information~\cite{2014Tinder}, spatial regions~\cite{zhao2018exploiting}, or user trajectories~\cite{yadav2024protecting,wang2023location} to recover sensitive location data.

However, location privacy threats in \emph{k}-nearest neighbor query (kNNQ) services (e.g., ``nearby people'' features) have been largely overlooked. To the best of our knowledge, location inference attacks specifically targeting kNNQ services have not been systematically studied. In such systems, attackers may exploit the ranked neighbor list returned by the API to infer the precise location of a target user. Given the widespread use of kNNQ in social and recommendation applications (such as Tinder~\cite{2014Tinder}), this work systematically investigates location privacy risks in kNNQ services and proposes an effective defense mechanism to mitigate such threats.

\noindent\textbf{Attack design and evaluation.}
We propose two location inference attacks against kNNQ services: the \emph{Geometric Intersection-based LIA} (GI-LIA) and the \emph{Zero-order Optimization-based LIA} (ZO-LIA). 
GI-LIA leverages distance relationships between the attacker at multiple query locations and the target, and solves the target's coordinates using geometric intersection. 
In contrast, ZO-LIA estimates a pseudo-gradient from the target's ranking information in query results and iteratively approximates the target's true location.  
Experimental results demonstrate that both attacks are highly effective, achieving inference success rates above $95\%$. 
These findings indicate that current kNNQ services face severe privacy risks even when users' exact coordinates are not explicitly disclosed.

\noindent\textbf{Design of defense method.}
To mitigate location privacy risks in kNNQ, we propose \emph{DPRS}, a privacy-preserving perturbation framework. 
DPRS perturbs original location data within a carefully constructed privacy region, thereby reducing distance distortion among nearest neighbors. 
To achieve this goal, we design a rejection-sampling-based location perturbation mechanism that repeatedly proposes and tests candidate locations within a specified interval to ensure the resulting samples follow the target distribution. 
This mechanism also constrains perturbed samples within the interval, preventing excessive deviation. Applying rejection sampling to DP-kNNQ poses two key challenges: guaranteeing differential privacy under the altered noise distribution and selecting an appropriate perturbation interval to preserve query utility. To address them, we characterize the post-rejection noise using Rényi divergence and design a private interval construction algorithm that adaptively determines the perturbation interval.

\noindent\textbf{Evaluation.}
We conduct extensive experiments on two real-world datasets and two synthetic datasets to evaluate the performance of DPRS. 
The results show that DPRS effectively protects location privacy and significantly reduces privacy leakage. 
In particular, under a low privacy budget, DPRS reduces the success rates of both GI-LIA and ZO-LIA attacks to below $3\%$, demonstrating the robustness of the proposed defense. 
Moreover, under the same privacy budget, DPRS consistently outperforms three state-of-the-art location differential privacy methods~\cite{andres2013geo,hong2022collecting,wang2022srr} across multiple utility metrics. The main contributions of this paper are summarized as follows:

\begin{itemize}
\item We present the first location inference attacks against kNNQ services, exposing serious privacy risks in existing kNN-based location services.
\item We propose DPRS, a privacy-preserving framework for kNNQ that achieves a better privacy--utility trade-off through constrained perturbation.
\item We establish an RDP-based privacy guarantee for the rejection sampling mechanism and derive its utility bound under the default interval.
\item Experiments on two real-world and two synthetic datasets show that DPRS consistently outperforms existing methods in both privacy protection and query utility.
\end{itemize}

\vspace{-2mm}
\section{Background and Related Work} \label{background}

Existing location inference attacks (LIAs) deduce true locations from obfuscated data using probabilistic models and contextual information~\cite{andres2013geo,yadav2024protecting,wang2023location,DBLP:conf/woot/BitsikasSPR24}, but fundamentally rely on access to coordinate data or trajectory. This work explores a more stringent scenario: the LBS only returns an anonymous, ranked list of k-nearest neighbors (kNN)~\cite{seidl1998optimal}. This operation, known as a k-nearest neighbor query (kNNQ), is a fundamental block in numerous modern applications and is defined as follows:

\begin{definition}({\bf K-nearest neighbor query~\cite{seidl1998optimal}}) \label{definition:kNNS}
For a query object $q \in O$ and a query parameter $k$, the k-nearest neighbor query returns the smallest set $NN_{q}(k) \subseteq D$ that contains (at least) $k$ objects from the database, satisfying:
\begin{align}
NN_{q}(k) = \{ o \in D \mid d(o, q) \le d(o_k, q) \},
\end{align}
where $o_k$ is the k-th closest object to $q$ in the database $D$.
\end{definition}

Although kNNQ seemingly protects privacy by concealing coordinates, its ranked list forms a potent side-channel highly susceptible to location inference. Differential privacy (DP)~\cite{dwork2014algorithmic} is currently the gold standard for geolocation protection.

\begin{definition}
({\bf Differential Privacy~\cite{dwork2014algorithmic}}). The algorithm $A$ provides ($\epsilon$, $\delta$)-Differential Privacy, if for any two neighboring datasets $D$ and $D'$ that differ in only a single entry, $\forall S \subseteq \text{Range}(A)$, the following holds:
\begin{equation}
{\rm Pr}(A(D) \in S) \leq e^{\epsilon} \times {\rm Pr}(A(D') \in S) + \delta.
\end{equation}
\end{definition}

DP, foundationalized by geo-indistinguishability~\cite{andres2013geo}, is widely applied in LBS \cite{wang2021accurate,niu2020eclipse}. However, existing DP-kNNQ methods suffer from severe flaws: L-SRR~\cite{wang2022srr} can only perform kNNQ by perturbing users to a finite, predefined set of points of interest (POIs). This approach severely distorts the data's underlying spatial distribution, rendering the dataset unusable for other aggregation analyses. Meanwhile, other methods, such as ReuseKNN \cite{mullner2023reuseknn}, provide uneven protection; they only apply perturbation to a small subset of users deemed vulnerable, leaving the broader user population with weaker privacy guarantees. Schemes like FedKNN \cite{zhang2024fedknn}, on the other hand, only perturb each participant's contribution weight, rendering their security entirely dependent on hardware isolation and creating a catastrophic ``all-or-nothing'' privacy risk. To overcome these limitations, we propose the DPRS mechanism. In this paper, we use a more powerful privacy framework: Rényi Differential Privacy (RDP) \cite{IlyaMironov2017RnyiDP} for privacy loss analysis.
\begin{definition}({\bf Rényi Divergence \cite{van2014renyi}}) Given two probability distributions $P$ and $Q$, the Rényi divergence of order $\alpha > 1$ is: 
\begin{align}
D_{\alpha}(P \| Q)=\frac{1}{\alpha-1} \ln \mathbf{E}_{x \sim Q}\left[\left(\frac{P(x)}{Q(x)}\right)^{\alpha}\right],
\end{align}
where $P(x)$, $Q(x)$ denotes the probability density of $P$, $Q$.
\end{definition}

\vspace{-2mm}
\section{Problem Setup} \label{problem setup}

\begin{figure}[tbp] 
    \centering 
    \includegraphics[width=0.7\columnwidth]{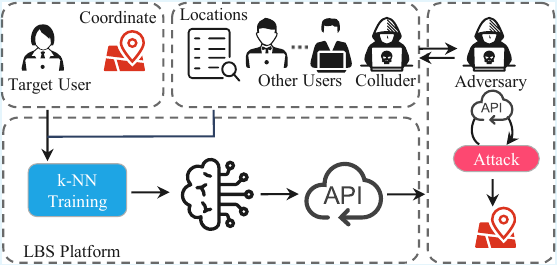} 
    \caption{Threat model}
    \label{fig:setup} 
\end{figure}

We consider a location-based service (LBS) platform that maintains a set of users 
$\mathcal{U} = \{u_1, u_2, ..., u_N\}$, where each user $u_i$ is associated with a 
2D geographic coordinate $L_i = (x_i, y_i)$. 
The LBS server supports the kNNQ interface that allows 
a user to submit a query location $L_q$ and obtain a ranked list of the $k$ nearest users 
based on their distances to $L_q$. Formally, given a query location $L_q$, the server returns an ordered list: 

\begin{align}
\text{kNNQ}(L_q) = [u_{(1)}, u_{(2)}, ..., u_{(k)}].
\end{align}

To protect user privacy, the server does not reveal the exact distance values or 
the precise locations of users. Instead, it only returns the ranked order of users.

\vspace{-1mm}
\subsection{Threat Model}
We consider the adversary who is a legitimate registered user of the LBS 
platform and has normal access to the kNNQ API. The adversary cannot compromise the 
server or access its internal data structures, and can only interact with the system 
through API queries. The adversary possesses the following capabilities:

\begin{itemize}

\item \textbf{API Access.}
The adversary can repeatedly issue kNN queries from arbitrary locations and observe the returned ranked list of anonymized user IDs.

\item \textbf{Collusion.}
The adversary can collaborate with one or more colluding users whose precise 
locations are known to the adversary. These colluders act as mobile probes that can 
move to different locations and report their exact coordinates.

\item \textbf{Adaptive Queries.}
The adversary can adaptively change the probe locations and repeatedly query the API 
to observe how the ranking of the target user changes with respect to other users.

\end{itemize}

Figure~\ref{fig:setup} illustrates the threat model. Although the server does not reveal distance values, the ranking output implicitly 
encodes relative distance information. 
In particular, when a probe moves to a location where the rank of the target user 
changes relative to another user, this event indicates that the probe lies on a 
boundary where the distances to the two users become equal.

\noindent \textbf{Practicality of the Attack.}
This attack scenario naturally arises in many real-world location-based services. 
To protect user privacy, many LBS platforms avoid revealing exact locations or 
distances and instead return a ranked list of nearby users. For example, social 
applications such as Tinder~\cite{2014Tinder} provide “nearby users” functionality that displays 
profiles ordered by proximity while hiding precise coordinates or distance values. 
Since such services are accessible through mobile interfaces or public APIs, an adversary can repeatedly issue queries from different locations—either by moving 
their own device or collaborating with colluding users, and observe how the target 
user’s rank changes. These observations enable the adversary to gradually infer 
the target user’s location.

\vspace{-1mm}
\subsection{Attack Formulation}

Let the true location of the target user be $L_t = (x_t, y_t)$. The adversary interacts with the LBS server by issuing a sequence of kNN queries from different probe locations. Let
$Q = \{q_1, q_2, ..., q_n\}$ denotes the set of query results obtained from the server, where each query result 
$q_i$ corresponds to the ranked list returned by the kNNQ interface at probe 
location $L_i$. Meanwhile, the adversary also maintains a set of colluding probe locations $C = \{c_1, c_2, ..., c_l\}$, where each $c_j$ represents the exact coordinate of a colluder used in the probing 
process.

Through adaptive probing, the adversary observes the positions where the ranking of 
the target user changes relative to other users. Based on the collected query results and probe locations, the adversary aims to design a location inference algorithm.

\begin{align}
\mathcal{A}: (Q, C) \rightarrow L_{inf},
\end{align}
which outputs an estimated location of the target user: $L_{inf} = (x_{inf}, y_{inf})$.

\vspace{-1mm}
\subsection{Measurement of the Attack}\label{subsec:Measurement of the Attack}
To quantitatively evaluate the effectiveness and efficiency of the proposed attack, we adopt three metrics.

\noindent\textbf{Error Distance (Dist).} 
We measure the average inference error between the estimated location and the true location of the target user. 
Given a set of $N$ attack instances, the average error distance is defined as below: 

\begin{align}
\text{Dist} = \frac{1}{N} \sum_{i=1}^{N} d\big(L_{inf}^{(i)}, L_t^{(i)}\big),
\end{align}where $d(\cdot)$ denotes the Euclidean distance, $L_{inf}^{(i)}$ is the inferred location, and $L_t^{(i)}$ is the true location of the $i$-th target user.

\noindent\textbf{Attack Success Rate (Acc).} We define a successful attack as one where the inferred location lies within a predefined distance threshold $\tau$ from the true location. 
The attack success rate is defined as follows:
\begin{align}
\text{Acc} = \frac{1}{N} \sum_{i=1}^{N} \mathbf{1}\big[d(L_{inf}^{(i)}, L_t^{(i)}) \le \tau\big],
\end{align} where $\mathbf{1}[\cdot]$ is the indicator function. In our experimental evaluation, we set $\tau=100$ meters by default, following the commonly adopted evaluation standard in location privacy research~\cite{andres2013geo}.

\noindent\textbf{Attack Cost (Cost).}
To measure the efficiency of the attack, we compute the average time required to complete an attack instance: 
\begin{align}
\text{Cost} = \frac{1}{N} \sum_{i=1}^{N} T_i,
\end{align} where $T_i$ denotes the time (seconds) to infer the location of the $i$-th target user.

Among these metrics, \textbf{Dist} and \textbf{Acc} reflect the \emph{attack effectiveness}, while \textbf{Cost} measures the \emph{attack efficiency}.

\vspace{-2mm}
\section{Our Attack Methods} \label{LIA}
In this section, we present two attack methods against KNNQ. We first introduce a geometry-based attack and analyze its inherent limitations in terms of communication overhead. We then propose a more communication-efficient attack framework that leverages zeroth-order optimization to achieve a better trade-off between attack effectiveness and efficiency.

\vspace{-1mm}
\subsection{A Straw-man Approach: GI-LIA}

The first approach to location inference is based on geometric triangulation. We refer to this baseline method as the Geometric Intersection Location Inference Attack (GI-LIA). The key idea is to infer the target's location by determining the intersection points of two circular trajectories on which the target must reside.

First, the attacker establishes an initial circular trajectory centered at an attack point $A_1$. Specifically, the attacker determines the radius $R_1$ of the circle $C(A_1, R_1)$ through a binary search procedure assisted by a collaborator. The collaborator moves radially with respect to $A_1$, while the attacker repeatedly issues kNNQ queries. By adjusting the collaborator's distance from $A_1$, the attacker identifies a distance $R_1$ at which the collaborator's rank in the query result becomes exactly $k$. Since the target also appears at rank $k$, the target must lie on the circumference of $C(A_1, R_1)$.

Next, the attacker selects a second attack point $A_2$ to obtain another geometric constraint. The attacker performs a search over candidate locations for $A_2$, ensuring that (i) $A_2$ is sufficiently distant from $A_1$ to provide geometric stability and (ii) the target remains within the kNNQ result set. Once such a point is found, the attacker repeats the binary search procedure to determine the second circular trajectory $C(A_2, R_2)$. The intersection of the two circles $C(A_1, R_1)$ and $C(A_2, R_2)$ yields two candidate locations $\{P_1, P_2\}$. The attacker then issues verification queries at these candidate points and selects the location that yields the lower rank for the target as the inferred target position. The total communication overhead of GI-LIA can be expressed as:

\begin{equation}
Cost_{GI-LIA} =
Cost_{binary}^{R_1}
+
Cost_{search}(A_2)
+
Cost_{binary}^{R_2}
+
2,
\end{equation}

where $Cost_{binary}$ denotes the query cost of a single binary search used to determine a circle radius, and $Cost_{search}(A_2)$ represents the cost of locating the second attack point $A_2$. The final term accounts for the two verification queries at candidate locations $P_1$ and $P_2$. Because both the binary search procedures and the search for $A_2$ require numerous queries, GI-LIA incurs substantial communication overhead, making it inefficient and impractical for real-time or stealthy attacks.

\begin{figure}[!t] 
    \centering 
    \includegraphics[width=0.75\columnwidth]{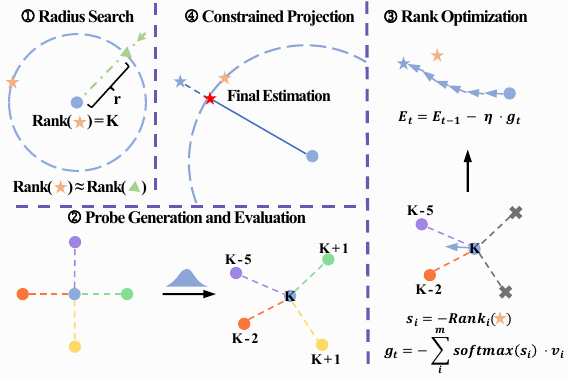} 
    \caption{The framework of ZO-LIA}
    \label{fig:zoa} 
\end{figure}

\vspace{-1mm}
\subsection{Improvement: ZO-LIA}

To reduce the excessive communication overhead of GI-LIA, we propose an efficient \textbf{zeroth-order optimization location inference attack (ZO-LIA)}. Unlike the GI-LIA, ZO-LIA avoids constructing a second geometric constraint. Instead, it performs a rank-guided local search within the feasible region defined by the first constraint circle and then projects the result onto the circle boundary. The overall process is illustrated in Figure~\ref{fig:zoa}.

\noindent \textbf{Phase \ding{172}: Radius Search.}  
ZO-LIA first determines the initial constraint circle $C(A_1, R_1)$ using the same radius search procedure described in GI-LIA. By performing a binary search with the help of a collaborator, the attacker finds a distance $R_1$ such that the collaborator appears at rank $k$ in the kNNQ results. Since the target also appears at rank $k$, the target must lie on the circumference of $C(A_1, R_1)$.

\noindent\textbf{Phase \ding{173}: Probe Generation and Evaluation.}  
Instead of constructing a second circle, ZO-LIA generates a set of probe points around the circle center $A_1$. These probes are sampled along multiple directions in the local neighborhood. For each probe point, the attacker issues a kNNQ query and records the target's rank. The rank serves as a score indicating how promising the corresponding direction is for approaching the target.

\noindent\textbf{Phase \ding{174}: Rank Optimization.}  
This phase constitutes the core of ZO-LIA. We employ a zeroth-order optimization strategy that relies solely on rank feedback. Using the rank scores obtained from the probe evaluations, the attacker estimates a pseudo-gradient that approximates the direction in which the target's rank decreases most rapidly. The attacker's position is then iteratively updated along this estimated direction, gradually moving toward an internal point $P_{opt}$ that minimizes the observed rank of the target.

\noindent\textbf{Phase \ding{175}: Constrained Projection.}  
Once the search converges to the internal point $P_{opt}$, a final constrained projection step is applied. Because the target must lie on the circumference of $C(A_1, R_1)$, the attacker orthogonally projects $P_{opt}$ onto the circle boundary. This projection is a simple local computation with constant complexity and produces the final estimate of the target's location. The total communication overhead of ZO-LIA is given by

\begin{equation}
Cost_{ZO-LIA} = Cost_{binary} + N_{iter} \cdot Cost_{0-iter},
\end{equation}

where $Cost_{binary}$ denotes the query cost of the initial radius search, $N_{iter}$ is the number of optimization iterations, and $Cost_{0-iter}$ represents the number of probe queries per iteration. 

\vspace{-1mm}
\subsection{Attack performance in real-world scenarios} \label{sub:Attack performance in real-world scenarios}
We evaluate our attacks on two real-world and two synthetic datasets. 
The evaluation metrics include \textit{Dist}, \textit{Acc}, and \textit{Cost} (see~\Sref{subsec:Measurement of the Attack}). 
We consider $k=10$, $k=30$, and $k=50$. 
For each $k$, we perform 5 independent runs with 50 attack instances per run, and report the average results. The results are shown in Table~\ref{tab:attack_performance_with_comm_final}. Other attack setups are provided in Appendix \ref{setup of LIA}.

\begin{table*}[t]
\centering
\caption{The performance of ZO-LIA and GI-LIA methods across four datasets. }
\label{tab:attack_performance_with_comm_final}
\setlength{\tabcolsep}{3pt}

\begin{tabular}{cl ccc ccc ccc}
\toprule
\multirow{2}{*}{\textbf{Method}} & \multirow{2}{*}{\textbf{Dataset}} & \multicolumn{3}{c}{\textbf{k=10}} & \multicolumn{3}{c}{\textbf{k=30}} & \multicolumn{3}{c}{\textbf{k=50}} \\
\cmidrule(lr){3-5} \cmidrule(lr){6-8} \cmidrule(lr){9-11}
& & \textbf{Dist} & \textbf{Acc} & \textbf{Cost} & \textbf{Dist} & \textbf{Acc} & \textbf{Cost} & \textbf{Dist} & \textbf{Acc} & \textbf{Cost} \\
\midrule

\multirow{4}{*}{\textbf{ZO-LIA}}
& brightkite & 28.91 & 0.948 & 15.94 & 40.59 & 0.914 & 17.62 & 41.98 & 0.912 & 19.21 \\
& gowalla    & 39.20 & 0.960 & 16.05 & 56.03 & 0.906 & 17.86 & 56.83 & 0.926 & 18.92 \\
& gaussian   & 24.35 & 0.956 & 22.14 & 42.17 & 0.910 & 23.21 & 43.56 & 0.893 & 24.08 \\
& beta       & 40.39 & 0.944 & 22.17 & 43.96 & 0.896 & 23.09 & 48.51 & 0.873 & 23.96 \\
\midrule

\multirow{4}{*}{\textbf{GI-LIA}}
& brightkite & 26.93 & 0.972 & 31.86 & 32.67 & 0.984 & 33.07 & 25.74 & 0.996 & 35.12 \\
& gowalla    & 20.59 & 0.996 & 31.42 & 24.55 & 0.992 & 34.67 & 23.56 & 0.992 & 35.98 \\
& gaussian   & 26.31 & 0.940 & 44.42 & 18.34 & 0.968 & 45.97 & 29.98 & 0.964 & 48.13 \\
& beta       & 39.96 & 0.948 & 43.21 & 41.67 & 0.948 & 45.07 & 51.84 & 0.976 & 47.31 \\
\bottomrule
\end{tabular}%

\end{table*}

\noindent\textbf{Attack effectiveness.}
The results demonstrate that both attacks achieve consistently high effectiveness across all datasets. In particular, both GI-LIA and ZO-LIA maintain success rates above 90\% under most settings, indicating that kNNQ-based LBS systems are highly vulnerable to location inference attacks. Among the two methods, GI-LIA achieves the best attack performance. It consistently yields the lowest average distance error and the highest success rates across datasets and k-values. For example, on the Brightkite dataset with $k=10$, GI-LIA achieves a success rate of 0.972 with an average distance error of only 26.93 meters. These results confirm that the geometric intersection strategy can accurately recover the target’s location.

\noindent\textbf{Attack efficiency.}
Although GI-LIA achieves the highest attack accuracy, it incurs significantly higher communication overhead. In contrast, ZO-LIA substantially reduces the attack cost while maintaining high success rates. As shown in Table~\ref{tab:attack_performance_with_comm_final}, ZO-LIA consistently requires roughly half the execution time of GI-LIA across all datasets and k-values. For instance, on the Brightkite dataset with $k=30$, the time overhead of ZO-LIA is 17.62 seconds, compared to 33.07 seconds for GI-LIA. This reduction implies fewer queries are needed, making ZO-LIA faster, stealthier, and more practical for real-world attacks.

\vspace{-2mm}
\section{Proposed Defense Method: DPRS} \label{defence method}
To protect location privacy under kNN queries, we propose the DPRS mechanism. DPRS perturbs the geographic locations in the dataset to produce a new location dataset that strictly satisfies differential privacy guarantees. As illustrated in Figure~\ref{fig:overall} and Algorithm \ref{algo:main_prism}, DPRS consists of two main components: (i) \textbf{Private Interval Construction (PIC):} Establishes a private, data-aware interval for the perturbation. (ii) \textbf{Reject Sampling Mechanism (RSM):} The core engine that efficiently perturbs data into the private interval constructed by the previous module.

\begin{figure*}[hbt]
    \centering 
    \includegraphics[width=0.9\textwidth]{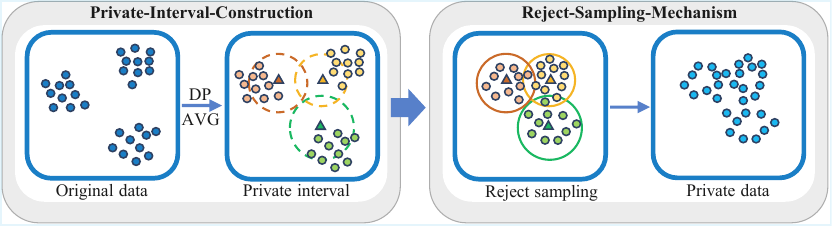}
    \caption{Overall framework of DPRS}
    \label{fig:overall}
\end{figure*}

\vspace{-2mm}
\subsection{Private Interval Construction} \label{construct interval}

\begin{algorithm}[t]
\caption{Private Interval Construction (PIC)}\label{algo:establishing centers}
\KwIn{dataset $D$, sensitivity $2r+1$, number of clusters $m$, number of iterations $N_c$, noise parameter of Lap distribution $\lambda_c$, radius scaling $\gamma$.}

\KwOut{$m$ private centroids.}

Initialize $k$ centroids by initial centroid selection.

\For{iter $\xleftarrow{} 1$ to $N_c$}{

    Get $m$ clusters through the standard k-means.
    
    Recalculate the centroid of each cluster.
    
    \For{$j \xleftarrow{} 1$ to $m$}{
    
        \For{$i \xleftarrow{} 1$ to $d$}{
        
            $sum'(C_j^*)[i] = sum(C_j^*)[i] + Lap(0,\lambda_c)$.
        }
        $num'(C_j^*) = num(C_j^*) + Lap(0,\lambda_c)$.
        
        $c_j = \frac{sum'(C_j^*)}{num'(C_j^*)}$.
    }  
}

\For{$j \xleftarrow{} 1$ to $m$}{
    Find the nearest centroid $c_{min}$ from all centroids excluding $c_j$.
    
    $R_j = \gamma \cdot distance(c_j, c_{min})$.
    
    $I_j = (c_j, R_j)$.
}

\Return $m$ private intervals $\{I_1, I_2, \dots, I_m\}$.
\end{algorithm}

The privacy interval is designed as a Euclidean ball, consisting of two components: a private center and a private radius. Ideally, to avoid introducing excessive noise offset, the private center should be as close as possible to the original sample. However, directly perturbing the original sample to serve as the interval center would introduce significant distance errors. To address this issue, the DP-k-means algorithm is introduced to obtain more stable and differentially private centers by injecting noise during the aggregation stage. In addition, the private radius is determined by the distance between the two nearest private centers, thereby eliminating the need for additional privacy budget.

As shown in Algorithm~\ref{algo:establishing centers}, PIC is implemented as follows. In each iteration, after samples are assigned to their clusters, the sensitive centroid recomputation step is protected. Specifically, the coordinate sums and the member counts of each cluster are concatenated, and calibrated Laplace noise is injected. The combined operation has a total $L_1$ sensitivity of $2r+1$. The private radius of each centroid is then set to half of the distance to its nearest neighboring centroid, which naturally defines the corresponding privacy interval.

\vspace{-1mm}
\subsection{Reject Sampling Mechanism} \label{rsm section}
\begin{algorithm}[!t]
    \caption{Rejection Sampling Mechanism (RSM)}
    \label{alg:rejection_sampling_single}
    
    \SetKwInOut{KwIn}{Input}
    \SetKwInOut{KwOut}{Output}
    
    \KwIn{
        data $\mathbf{x}$, norm $n$, interval $I$ (center $\mathbf{x}_c$, radius $R$), sensitivity $2r$, Lap distribution $p(\cdot|\mathbf{x},\lambda_p)$, target function $f(\cdot)\propto p(\cdot|\mathbf{x},\lambda_p)$.
    }
    \KwOut{
        perturbed data $\mathbf{x}'$.
    }
    \If{$\|\mathbf{x} - \mathbf{x}_c\|_k \le R$}{
    
        $M \leftarrow 1$.
        
    }
    \Else{
        $B \leftarrow \text{GetBoundary}(\mathbf{x}_c,R)$.
        
        $\mathbf{z}^* \leftarrow \min_{\mathbf{z} \in B} \|\mathbf{z} - \mathbf{x}\|_n$.
        
        $M \leftarrow f(\mathbf{z}^*)$.
    }
    \While{true}{
    
        $u,u_1,u_2\leftarrow \text{RandomUniform}(0, 1)$.
        
        $R \leftarrow R \sqrt{u_1}$, $\theta \leftarrow 2 \pi u_2$.
        
        $\mathbf{x}' \leftarrow \mathbf{x}_c + (R \cos(\theta),R \sin(\theta))$.
        
        \If{$u < \frac{f(\mathbf{x}')}{M}$}{
        
            \Return{$\mathbf{x}'$}.
        }
    }
    
\end{algorithm}

The rejection sampling mechanism efficiently perturbs private data into the privacy interval constructed by the PIC module. We begin by specifying a base noise distribution, typically $p(\cdot)=\text{Lap}(\cdot \mid \mathbf{x},\lambda)$. 
By normalizing the 2D location data to the range $[-r,r]$, the $L_1$ sensitivity is bounded by $2r$, which enables calibrated noise injection. 
To perturb data within the specified Euclidean sphere, we define $f(\cdot) \propto p(\cdot)$, thereby avoiding the computation of the normalization constant. Rejection sampling is then used to generate valid perturbed samples. 
This procedure requires an upper bound $M$ of $f(\cdot)$ within the sphere in order to compute the acceptance probability. 
As shown in Algorithm~\ref{alg:rejection_sampling_single} (Lines~1--6), a tight bound $M$ is obtained under either the $L_1$ norm (Laplace mechanism) or the $L_2$ norm (Gaussian mechanism) by evaluating a set of candidate points on the boundary of the sphere via $\text{GetBoundary}(\mathbf{x}_c, R)$ (Line 4). Furthermore, we also extend RSM to the Gaussian distribution in Appendix \ref{appendix:Gau}.

Given these components, the mechanism iteratively performs rejection sampling until a valid sample is accepted. 
Specifically, a candidate point $\mathbf{x}'$ is first drawn uniformly from the sphere. 
The acceptance ratio $f(\mathbf{x}')/M$ is then computed, followed by drawing a random variable $u \sim U(0,1)$. 
The candidate is accepted if $u \le f(\mathbf{x}')/M$; otherwise, the sample is rejected and the process repeats. Additionally, Theorem \ref{definition:rejection_sampling_single} and Theorem \ref{thm:utility-lap-truncated-2d} in the Appendix \ref{subsec:Lemma theorem supplementation} guarantee the distributional correctness of the rejection sampling mechanism and provide an upper bound on utility.

\vspace{-2mm}
\section{Privacy Analysis}\label{privacy analysis}

 We first analyze the privacy of the PIC (Algorithm~\ref{algo:establishing centers}) and RSM modules (Algorithm~\ref{alg:rejection_sampling_single}), then present the total privacy budget of DPRS (Algorithm \ref{algo:main_prism}).
\begin{algorithm}[t]

 \SetAlgoLined
 \DontPrintSemicolon
 \caption{Overall algorithm of DPRS}
 \label{algo:main_prism}

 \SetKwInOut{KwIn}{Input}
 \SetKwInOut{KwOut}{Output}
 
 \KwIn{
    dataset $\mathcal{D}$, cluster number $m$, iteration $N_c$, noise parameter $\lambda_p,\lambda_c$.
 }
 \KwOut{perturbed dataset $\mathcal{D}'$.}

 $\{ I_1, I_2, \dots, I_k \} \leftarrow \text{PIC} (\mathcal{D}, m, N_c, \lambda_c)$.

 Partition $\mathcal{D}$ into $m$ clusters $\{ C_1, C_2, \dots, C_m \}$.

 $\mathcal{D}' \leftarrow \emptyset$.
 
 \For{$j \leftarrow 1$ \KwTo $m$}{
 
    \For{$\mathbf{x} \in C_j$}{
    
        $\mathbf{x}' \leftarrow \text{RSM}(\mathbf{x}, I_j, \lambda_p)$.
        
        $\mathcal{D}' \leftarrow \mathcal{D}' \cup \{\mathbf{x}'\}$.
    }
 }
 \Return{$\mathcal{D}'$}.
 
\end{algorithm}

\subsection{Privacy Analysis of DPRS Components}

\begin{theorem} \label{definition:establishing centers}
Algorithm \ref{algo:establishing centers} satisfies $(\alpha, \frac{\epsilon_c}{\alpha - 1} )$-RDP, where
\begin{align}
\epsilon_c =   
    & 3N_c\log \left(\frac{\alpha}{2 \alpha - 1} \exp \left( \frac{\alpha - 1}{\lambda_c} \right) + \frac{\alpha - 1}{2 \alpha - 1} \exp \left( \frac{-\alpha }{\lambda_c} \right)\right).
\end{align}
\end{theorem}

\begin{proof}
Let $D$ be a dataset partitioned into $k$ disjoint clusters, $C_1, \dots, C_k$. A neighboring dataset $D'$ is formed by removing a single point, such that $D' = D \setminus \{x\}$ for some $x \in D$, and the partition on $D'$ is composed of clusters $C'_1, \dots, C'_k$. There exists an index $J$ such that $C'_J = C_J \setminus \{x\}$ and $C'_j = C_j$ elsewhere.

Therefore, each iteration can be regarded as the parallel composition of mechanisms querying function $f(\cdot)$ over $k$ disjoint clusters, and differential privacy achievement is determined by the mechanism over cluster $C_J$ (since mechanisms over other clusters $C_j^{'}$ achieve 0-differential privacy for $C_j = C_{j}^{'}$).

Let $p(\cdot)$ and $p'(\cdot)$ denote the probability density functions of the mechanism over cluster $C_J$. For any point $v \in D^{2 \times N}$, its RDP characterization is formed by the composition of three single-dimension Laplace mechanisms, each with a sensitivity of 1, iterated $N_c$ times. The first two mechanisms correspond to the position dimensions, and the last one corresponds to the counting dimension. According to Lemma \ref{lem:composition} (\Sref{subsec:dprs analysis}), its privacy characterization is as follows:

\begin{align}
\begin{split}
\begin{aligned}
D_{\alpha}(p(v)||p'(v)) = \frac{3N_c}{\alpha-1}\log\left(\frac{\alpha}{2\alpha-1}\exp\left(\frac{\alpha-1}{\lambda}\right)+\frac{\alpha-1}{2\alpha-1}\exp\left(-\frac{\alpha}{\lambda}\right)\right).
\end{aligned}
\end{split}
\end{align}
Therefore, Theorem \ref{definition:establishing centers} is proved.

\end{proof}

\begin{theorem} \label{definition:selective publish of Laplace}
Algorithm~\ref{alg:rejection_sampling_single} satisfies $(\alpha, \frac{\epsilon_p}{\alpha - 1} )$-RDP, where
\begin{align}
\epsilon_p =   
    & 2\log\left(\frac{\alpha}{2 \alpha - 1} \exp \left( \frac{\alpha - 1}{\lambda_p} \right) + \frac{\alpha - 1}{2 \alpha - 1} \exp \left( \frac{-\alpha}{\lambda_p} \right)\right).
\end{align}

\end{theorem}

\begin{proof}
 let $D_{\alpha}(S)$ be RDP of Algrithm \ref{alg:rejection_sampling_single},  according to Lemma \ref{theorem:rdp_new_feature} (Appendix \ref{subsec:Lemma theorem supplementation}) we have:    
\begin{align}
\begin{split}
\begin{aligned}
 D_{\alpha}(S) \leq D_{\alpha}(Lap(\mathbf{t}|\mathbf{0},\lambda)||Lap(\mathbf{t}|\mathbf{\mu},\lambda)).
\end{aligned}
\end{split}
\end{align}
In fact, the Rényi divergence of the multidimensional Laplace distribution has no analytical expression. Therefore, we characterize the Rényi differential privacy of the multidimensional Laplace distribution as a combination of different dimensions. According to Lemma \ref{lem:composition} (\Sref{subsec:dprs analysis}), we have:
\begin{align}
\begin{split}
\begin{aligned}
D_{\alpha}(Lap(\mathbf{t}|\mathbf{0},\lambda)||Lap(\mathbf{t}|\mathbf{\mu},\lambda)) = 2 D_{\alpha}(\Lambda(0,\lambda)||\Lambda(\mu,\lambda)).
\end{aligned}
\end{split}
\end{align}
Therefore, we have:
\begin{align}
\begin{split}
\begin{aligned}
&D_{\alpha}(S) \\
&\leq\frac{2}{\alpha-1}\log\int_{-\infty}^\infty\left(\frac{1}{2\lambda}e^{-\frac{|x-\mu|}{\lambda}}\right)^\alpha\left(\frac{1}{2\lambda}e^{-\frac{|x|}{\lambda}}\right)^{1-\alpha}dx\\
& = \frac{2}{\alpha-1}\log\left(\frac{1}{2\lambda}\int_{-\infty}^{\infty}\exp\left(-\frac{\alpha|x-\mu|-(\alpha-1)|x|}{\lambda}\right)dx\right)\\
&=\frac{2}{\alpha-1}\log \Bigg(\frac{1}{2\lambda}\Bigg(\int_{-\infty}^0e^{-\frac{\alpha\mu-x}{\lambda}}dx+ \int_0^\mu e^{-\frac{\alpha\mu-(2\alpha-1)x}{\lambda}}dx + \int_\mu^\infty e^{-\frac{x-\alpha\mu}{\lambda}}dx\Bigg)\Bigg)\\
&= \frac{2}{\alpha-1}\log \Bigg(\left(\frac{1}{2}-\frac{1}{2(2\alpha-1)}\right)e^{-\frac{a\mu}{\lambda}}+\left(\frac{1}{2}+\frac{1}{2(2\alpha-1)}\right)e^{\frac{(a-1)\mu}{\lambda}}\Bigg)\\
& =\frac{2}{\alpha-1}\log\Bigg(\frac{\alpha}{2\alpha-1}\exp\left(\frac{(\alpha-1)\mu}{\lambda}\right)+\frac{\alpha-1}{2\alpha-1}\exp\left(-\frac{\alpha\mu}{\lambda}\right)\Bigg).
\end{aligned}
\end{split}
\end{align}
For single-dimension Laplace noise, it is normalized to the range of $[-1, 1]$, so we set $\mu=1$. Therefore, Theorem \ref{definition:selective publish of Laplace} is proved.
\end{proof}

\vspace{-2mm}
\subsection{Overall Privacy Analysis of DPRS} \label{subsec:dprs analysis}
To derive the total privacy loss, we compose the guarantees from Theorem \ref{definition:establishing centers} and \ref{definition:selective publish of Laplace}. This relies on two standard lemmas from the RDP literature:

\begin{lemma}\label{lem:composition}
({\bf Composition Theorem \cite{IlyaMironov2017RnyiDP}}). Let two queries $f$, $g$ be $(\alpha,R_1)$ and $(\alpha,R_2)$-RDP respectively. Then their composition $(f,g)$ is $(\alpha,R_1+R_2)$-RDP.
\end{lemma}

\begin{lemma}\label{lem:conversion}
({\bf Conversion Theorem ~\cite{balle2020hypothesis}}). If  $f$ is an $(\alpha,R)$-RDP query ,then it satisfies$(R+\ln ((\alpha-1) / \alpha)-(\ln \delta+ \ln \alpha) /(\alpha-1), \delta)$-DP for any $0<\delta<1$.
\end{lemma}

By applying Lemma \ref{lem:composition} to the results of Theorem \ref{definition:establishing centers}, \ref{definition:selective publish of Laplace}, and then applying Lemma \ref{lem:conversion}, we arrive at the final privacy guarantee for the DPRS framework:
\begin{theorem} \label{definition:overall theorem}
Algorithm \ref{algo:main_prism} satisfies $(\epsilon_{Total},\delta)$-DP, where
\begin{align}
\epsilon_{Total} =   \frac{\epsilon_c+\epsilon_p}{\alpha-1} + \log(\frac{\alpha-1}{\alpha}) -\frac{\log(\delta)+\log(\alpha)}{\alpha-1}.
\end{align}
\end{theorem}

\vspace{-2mm}
\section{Evaluation} \label{evaluation}

In this section, we mainly examine the trade-off between defense effectiveness and KNNQ utility achieved by DPRS and the baselines in the KNNQ task, and further analyze the impact of key parameters on DPRS. There are additional experiment results that can be found in Appendix \ref{additional experiment}.

\begin{table}[t!]
\centering
\caption{Statistics of the datasets}
\label{tab:dataset_stats}
\resizebox{0.8\textwidth}{!}{
\begin{tabular}{ccccc}
\toprule
\textbf{Type} & \textbf{Dataset} & \textbf{Locations} & \textbf{Density} & \textbf{Dispersion Index} \\
\midrule
\multirow{2}{*}{Real-World} & Brightkite~\cite{cho2011friendship} & 18,409 & 11.76 & 0.26\% \\
 & Gowalla~\cite{cho2011friendship} & 17,738 & 11.33 & 0.32\% \\
\midrule 
\multirow{2}{*}{Synthetic} & Gaussian~\cite{10.5555/3540261.3542251} & 25,000 & 15.97 & 0.61\% \\
 & Beta~\cite{10.1145/3318464.3389700} & 25,000 & 15.97 & 0.70\% \\
\bottomrule
\end{tabular}
}
\end{table}

\vspace{-1mm}
\subsection{\textbf{Setup}}

We evaluate DPRS and three baselines (SRR~\cite{wang2022srr}, Square \cite{hong2022collecting} and Laplace~\cite{andres2013geo}) on four datasets  under various $k$-NN settings. There are two real-world datasets (Brighktite~\cite{cho2011friendship} and Gowalla~\cite{cho2011friendship}) and two synthetic datasets (Gaussian~\cite{10.5555/3540261.3542251} and Beta~\cite{10.1145/3318464.3389700}). We restrict the real-world datasets to the San Francisco region, with latitude in $[37.5, 37.9]$ and longitude in $[-122.6, -122.2]$. The synthetic datasets are generated from Gaussian $N(0,1)$ and Beta $B(2,5)$ distributions. Dataset details are shown in Table~\ref{tab:dataset_stats}. Detailed descriptions of evaluation metrics, baselines, and parameter settings are provided in Appendix \ref{Experiment Setup of DPRS}.

\begin{table*}[t]
\centering
\caption{Performance comparison of different defense mechanisms in the 10-NN task. \textcolor{blue}{Recall} and \textcolor{blue}{Ratio} measure query utility, where higher values indicate better query performance. \textcolor{red}{Acc} and \textcolor{red}{Dist} characterize the attack performance of ZO-LIA, where lower Acc and higher Dist indicate stronger defense effectiveness.}
\label{tab:main_results_wide_with_srr}
\setlength{\tabcolsep}{2pt} 
\resizebox{\textwidth}{!}{
\begin{tabular}{*{18}{c}}
\toprule

\multirow{2}{*}{\textbf{Dataset}} & \multirow{2}{*}{\textbf{Metric}} & \multicolumn{4}{c}{$\epsilon=0.5$} & \multicolumn{4}{c}{$\epsilon=1$} & \multicolumn{4}{c}{$\epsilon=3$} & \multicolumn{4}{c}{$\epsilon=5$} \\
\cmidrule(lr){3-6} \cmidrule(lr){7-10} \cmidrule(lr){11-14} \cmidrule(lr){15-18}
& & \textbf{DPRS} & \textbf{SRR} & \textbf{Square} & \textbf{Laplace} & \textbf{DPRS} & \textbf{SRR} & \textbf{Square} & \textbf{Laplace} & \textbf{DPRS} & \textbf{SRR} & \textbf{Square} & \textbf{Laplace} & \textbf{DPRS} & \textbf{SRR} & \textbf{Square} & \textbf{Laplace} \\
\midrule

\multirow{4}{*}{Brightkite}
 & \textcolor{blue}{Recall(↑)} & \textbf{0.434} & 0.382 & 0.389 & 0.353 & \textbf{0.464} & 0.423 & 0.414 & 0.398 & \textbf{0.472} & 0.436 & 0.440 & 0.416 & \textbf{0.489} & 0.447 & 0.448 & 0.423 \\
 & \textcolor{blue}{Ratio (↑)}  & \textbf{0.852} & 0.814 & 0.811 & 0.745 & \textbf{0.881} & 0.832 & 0.823 & 0.766 & \textbf{0.885} & 0.851 & 0.844 & 0.699 & \textbf{0.889} & 0.861 & 0.856 & 0.782 \\
\cmidrule{2-18}
 & \textcolor{red}{Acc (↓)}    & \textbf{0.012} & 0.016 & 0.018 & \textbf{0.012} & \textbf{0.014} & 0.020 & 0.018 & 0.016 & 0.020 & 0.024 & \textbf{0.018} & 0.022 & \textbf{0.022} & 0.028 & 0.024 & 0.026 \\
 & \textcolor{red}{Dist (↑)}   & 0.591 & 0.573 & 0.565 & \textbf{0.594} & \textbf{0.583} & 0.562 & 0.553 & 0.553 & 0.525 & 0.516 & \textbf{0.533} & 0.519 & \textbf{0.524} & 0.473 & 0.511 & 0.515 \\
\midrule

\multirow{4}{*}{Gowalla}
 & \textcolor{blue}{Recall(↑)} & \textbf{0.443} & 0.357 & 0.382 & 0.331 & \textbf{0.461} & 0.430 & 0.418 & 0.393 & \textbf{0.471} & 0.453 & 0.442 & 0.417 & \textbf{0.482} & 0.459 & 0.454 & 0.446 \\
 & \textcolor{blue}{Ratio (↑)}  & \textbf{0.827} & 0.773 & 0.779 & 0.705 & \textbf{0.872} & 0.830 & 0.815 & 0.737 & \textbf{0.879} & 0.844 & 0.840 & 0.749 & \textbf{0.886} & 0.853 & 0.848 & 0.797 \\
\cmidrule{2-18}
 & \textcolor{red}{Acc (↓)}    & 0.022 & \textbf{0.012} & 0.024 & 0.018 & \textbf{0.016} & 0.022 & 0.022 & 0.018 & \textbf{0.020} & 0.022 & 0.022 & \textbf{0.020} & 0.020 & 0.024 & \textbf{0.016} & 0.018 \\
 & \textcolor{red}{Dist (↑)}   & 0.553 & \textbf{0.587} & 0.544 & 0.572 & \textbf{0.530} & 0.481 & 0.499 & 0.517 & \textbf{0.517} & 0.486 & 0.476 & 0.497 & 0.412 & 0.409 & \textbf{0.476} & 0.446 \\
\midrule

\multirow{4}{*}{Gaussian}
 & \textcolor{blue}{Recall(↑)} & \textbf{0.644} & 0.612 & 0.587 & 0.541 & \textbf{0.655} & 0.637 & 0.624 & 0.607 & \textbf{0.664} & 0.643 & 0.630 & 0.611 & \textbf{0.672} & 0.649 & 0.645 & 0.629 \\
 & \textcolor{blue}{Ratio (↑)}  & \textbf{0.902} & 0.864 & 0.848 & 0.792 & \textbf{0.927} & 0.886 & 0.874 & 0.837 & \textbf{0.929} & 0.888 & 0.879 & 0.851 & \textbf{0.933} & 0.900 & 0.897 & 0.863 \\
\cmidrule{2-18}
 & \textcolor{red}{Acc (↓)}    & \textbf{0.008} & 0.014 & 0.018 & 0.012 & 0.020 & \textbf{0.016} & 0.022 & 0.020 & \textbf{0.014} & 0.016 & 0.018 & 0.018 & 0.018 & 0.022 & 0.020 & \textbf{0.016} \\
 & \textcolor{red}{Dist (↑)}   & \textbf{0.605} & 0.587 & 0.579 & 0.591 & 0.547 & \textbf{0.566} & 0.527 & 0.533 & \textbf{0.543} & 0.512 & 0.504 & 0.499 & 0.459 & 0.423 & 0.441 & \textbf{0.494} \\
\midrule

\multirow{4}{*}{Beta}
 & \textcolor{blue}{Recall(↑)} & \textbf{0.563} & 0.531 & 0.524 & 0.482 & \textbf{0.593} & 0.567 & 0.553 & 0.523 & \textbf{0.601} & 0.582 & 0.576 & 0.545 & \textbf{0.623} & 0.602 & 0.592 & 0.594 \\
 & \textcolor{blue}{Ratio (↑)}  & \textbf{0.827} & 0.751 & 0.734 & 0.678 & \textbf{0.861} & 0.782 & 0.768 & 0.714 & \textbf{0.873} & 0.807 & 0.796 & 0.749 & \textbf{0.886} & 0.822 & 0.815 & 0.782 \\
\cmidrule{2-18}
 & \textcolor{red}{Acc (↓)}    & \textbf{0.012} & 0.022 & 0.014 & 0.018 & 0.020 & 0.020 & 0.022 & \textbf{0.014} & 0.022 & 0.024 & \textbf{0.020} & 0.024 & \textbf{0.016} & 0.020 & 0.024 & \textbf{0.016} \\
 & \textcolor{red}{Dist (↑)}   & \textbf{0.616} & 0.586 & 0.608 & 0.597 & 0.564 & 0.553 & 0.576 & \textbf{0.599} & 0.512 & 0.494 & \textbf{0.529} & 0.481 & \textbf{0.574} & 0.517 & 0.483 & 0.564 \\
\bottomrule
\end{tabular}%
}
\end{table*}

\vspace{-2mm}
\subsection{\textbf{Tradeoff between Query Utility and Defense Effect}}

\textbf{Query utility.}
As shown in Table~\ref{tab:main_results_wide_with_srr}, DPRS consistently achieves the best query utility across all datasets and privacy budgets. It ranks first in both Recall and Ratio in all 16 settings. Averaged over all settings, DPRS attains a Recall of 0.546 and a Ratio of 0.882, outperforming SRR (0.513/0.835), Square (0.507/0.827), and Laplace (0.482/0.765). Compared with the strongest baseline in each setting, DPRS improves Recall and Ratio by 6.04\% and 5.33\% on average, respectively. For example, on Brightkite with $\epsilon=1$, DPRS improves Recall from 0.423 to 0.464 and Ratio from 0.832 to 0.881. On Beta with $\epsilon=1$, DPRS achieves a Ratio of 0.861, exceeding the best baseline by 0.08.

\noindent\textbf{Defense effect.}
DPRS also provides strong defense effectiveness against ZO-LIA \footnote[1]{As shown in Section~\ref{sub:Attack performance in real-world scenarios}, ZO-LIA achieves the strongest attack performance and is therefore used as the attack baseline.}. Averaged over all settings, it achieves the lowest attack success rate (Acc $=0.017$) and the largest distance error (Dist $=0.541$), compared with 0.020/0.520 for SRR, 0.020/0.525 for Square, and 0.018/0.536 for Laplace. In several cases, DPRS improves both utility and defense simultaneously. For instance, on Brightkite with $\epsilon=1$, DPRS reduces Acc from 0.02 to 0.014 and increases Dist from 0.562 to 0.583 compared with SRR. On Gaussian with $\epsilon=3$, DPRS achieves both the lowest Acc (0.014) and the highest Dist (0.543). Although DPRS is not always the best on every single attack metric, it remains highly competitive while preserving substantially better query utility.

\noindent\textbf{Summary.}
Overall, DPRS achieves the best trade-off between query utility and defense effect. It consistently dominates all baselines in Recall and Ratio, while also achieving the best average defense performance. These results verify that DPRS can effectively protect against location inference attacks without sacrificing kNNQ utility.

\begin{figure*}[t]
    \includegraphics[width=0.7\textwidth]{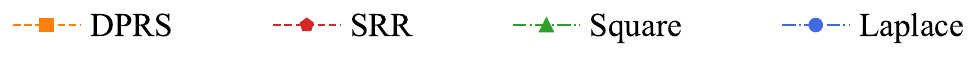}
  \centering
  
  \begin{subfigure}{0.26\linewidth}
    \centering
    \includegraphics[width=\linewidth]{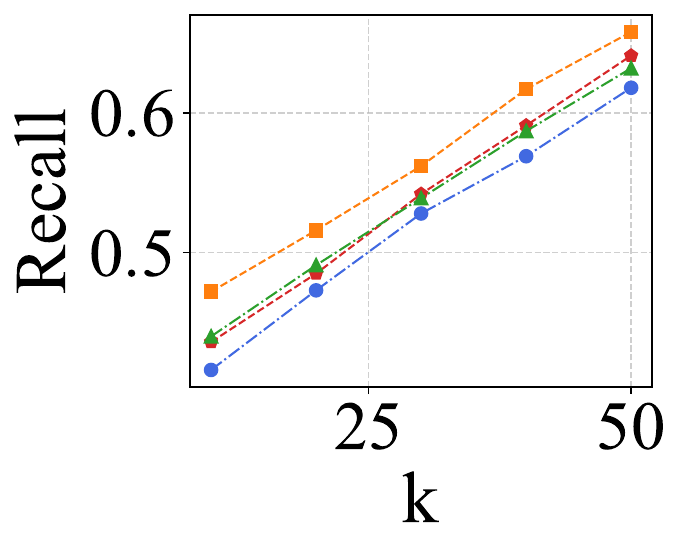}
    \captionsetup{skip=-2pt}
    \caption{Brightkite}
    \label{figure: k_ablation_brightkite}
  \end{subfigure}
  \begin{subfigure}{0.235\linewidth}
    \centering
    \includegraphics[width=\linewidth]{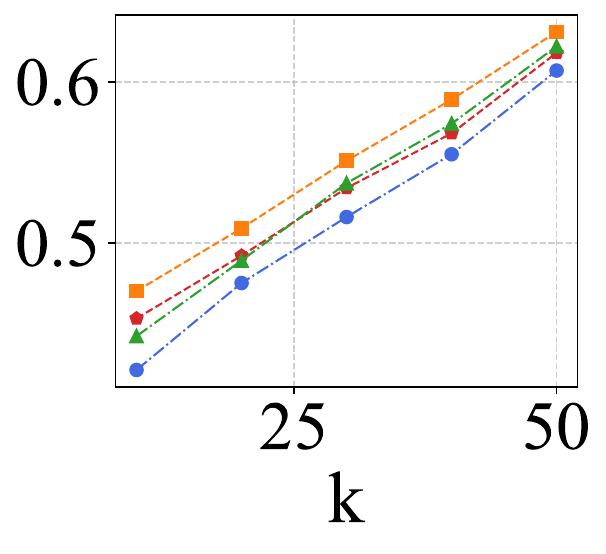}
    \captionsetup{skip=-2pt}
    \caption{Gowalla}
    \label{figure: k_variation_gowalla}
  \end{subfigure}
  \begin{subfigure}{0.235\linewidth}
    \centering
    \includegraphics[width=\linewidth]{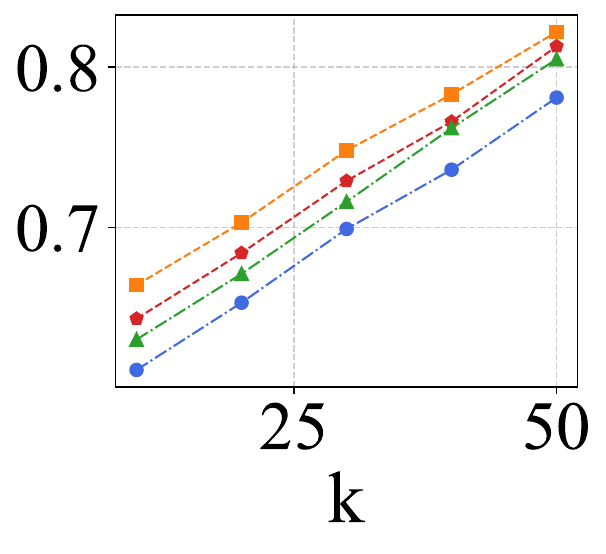}
    \captionsetup{skip=-2pt}
    \caption{Gaussian}
    \label{figure: k_variation_gaussian}
  \end{subfigure}
  \begin{subfigure}{0.235\linewidth}
    \centering
    \includegraphics[width=\linewidth]{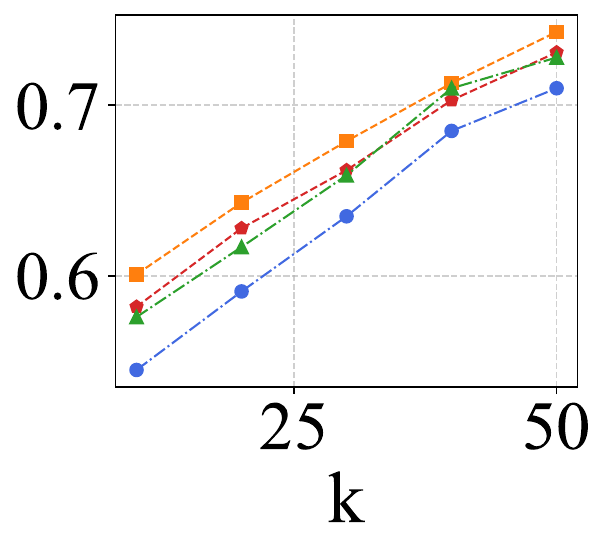}
    \captionsetup{skip=-2pt}
    \caption{Beta}
    \label{figure: k_variation_beta}
  \end{subfigure}
  \caption{Performance impact on the number of neighbors at $\epsilon=3$.}
  \label{figure: k_variation}
\end{figure*}

\begin{figure*}[t]
    \includegraphics[width=0.6\textwidth]{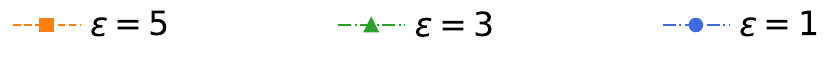}
  \centering

  \begin{subfigure}{0.26\linewidth}
    \centering
    \includegraphics[width=\linewidth]{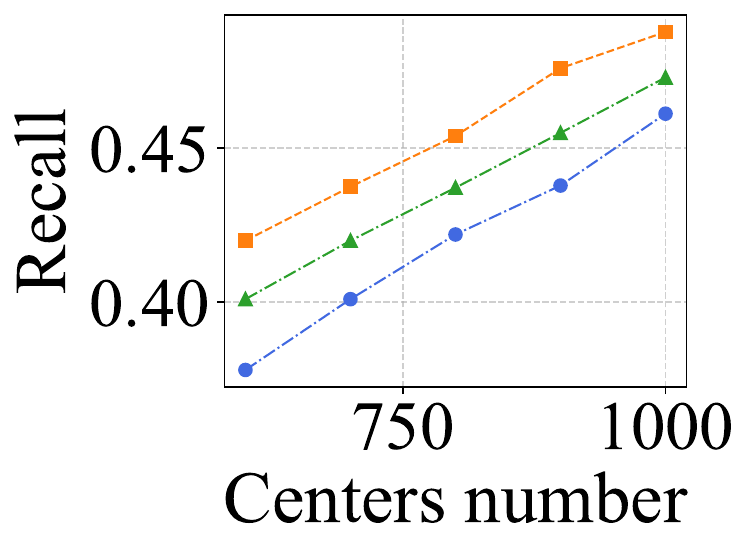}
    \captionsetup{skip=-2pt}
    \caption{Brightkite}
    \label{fig:span_sub1}
  \end{subfigure}
  \begin{subfigure}{0.235\linewidth}
    \centering
    \includegraphics[width=\linewidth]{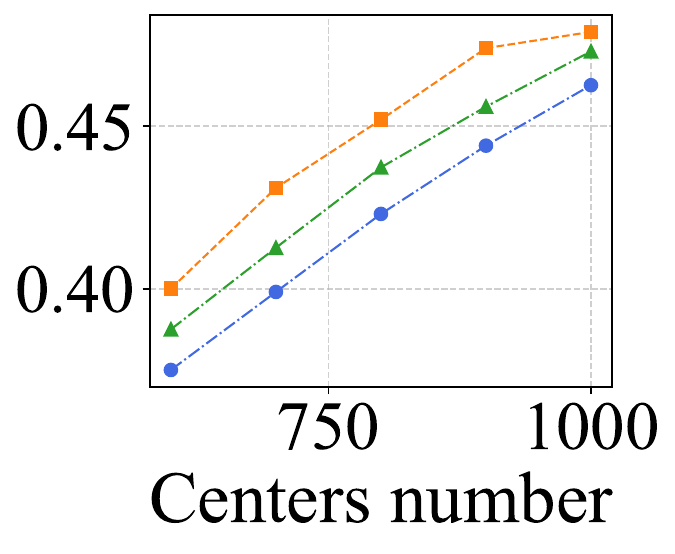}
    \captionsetup{skip=-2pt}
    \caption{Gowalla}
    \label{fig:span_sub2}
  \end{subfigure}
  \begin{subfigure}{0.235\linewidth}
    \centering
    \includegraphics[width=\linewidth]{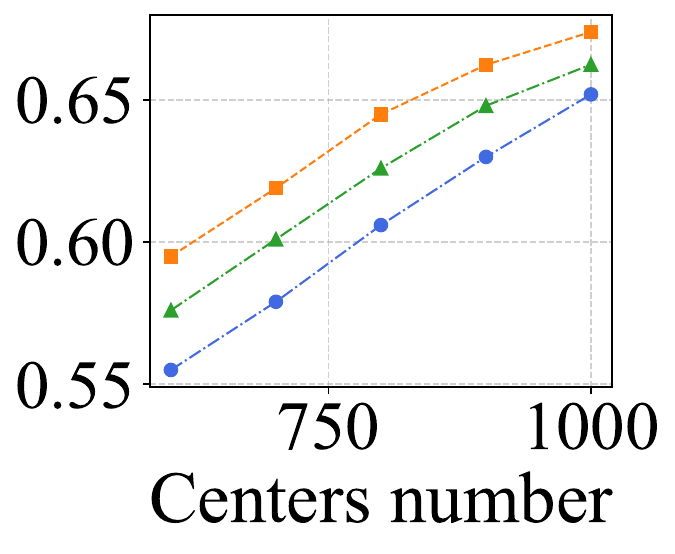}
    \captionsetup{skip=-2pt}
    \caption{Gaussian}
    \label{fig:span_sub3}
  \end{subfigure}
  \begin{subfigure}{0.225\linewidth}
    \centering
    \includegraphics[width=\linewidth]{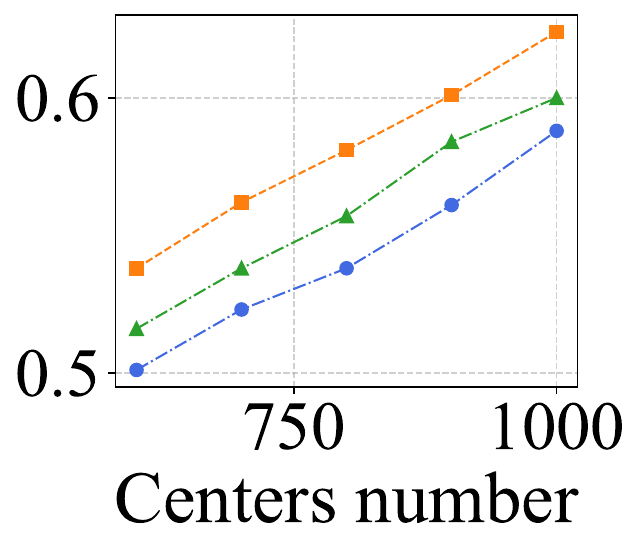}
    \captionsetup{skip=-2pt}
    \caption{Beta}
    \label{fig:span_sub4}
  \end{subfigure}

  \caption{Performance impact on the number of interval centers at $k=10$.}
  \label{fig:span_main}
\end{figure*}

\vspace{-2mm}
\subsection{\textbf{Parameter Sensitivity Analysis}}

\noindent\textbf{Impact of the Number of Neighbors.} We simulated scenarios with k ranging from 10 to 50 on two real-world datasets. As shown in Figure \ref{figure: k_variation}, we observe that the query recall consistently increases as the number of neighbors ($k$) grows. Furthermore, our method consistently outperforms the 2D Laplace, SRR, and Square mechanisms. This further demonstrates that the DPRS framework can effectively preserve the relative distances between neighboring samples.

\noindent\textbf{Impact of the Number of Centers.} Figure \ref{fig:span_main} illustrates the recall performance of our method as the number of interval centers increases from 600 to 1,000. We observe that the recall rate gradually rises with the number of perturbation centers. This is because more centers allow for a more detailed extraction of the dataset's features, thus capturing neighbor information at a finer granularity. 

\noindent\textbf{Impact of the Radius Scaling.} The size of the private interval directly determines the utility of the perturbed data, necessitating a careful trade-off for the radius scaling coefficient: An overly large coefficient leads to more frequent overlaps between different clusters, whereas a too small one disrupts the rank differences among nearest neighbors. As shown in Figure \ref{figure: threshold selection} in Appendix \ref{additional experiment}, setting the radius scaling coefficient to 0.5 is optimal.

\vspace{-1mm}
\subsection{Time Overhead Analysis}
DPRS avoids online query bottlenecks by performing PIC offline and periodically updating private intervals on LBS platforms. Furthermore, as shown in Figure \ref{figure: time overhead} in Appendix \ref{additional experiment}, our RSM incurs only lightweight online overhead given a private interval, which remains below 40ms and is substantially faster than standard noise-and-resampling methods. This demonstrates the scalability of DPRS to large-scale datasets.

\vspace{-2mm}
\section{Conclusion} \label{conclusion}
This work investigates the location privacy issues in kNNQ services. By proposing two attacks, GI-LIA and ZO-LIA, we reveal the privacy threats of kNNQ. To mitigate privacy risk, we propose DPRS, a differentially private framework tailored to high-sensitivity kNNQ. Both theoretical analysis and experimental results show that DPRS can effectively defend against LIA attacks while achieving better utility than state-of-the-art methods.

\vspace{-2mm}
\section*{Acknowledgment}
This work is supported by National Natural Science Foundation of China Key Program (Grant No. 62132005), Fundamental Research Funds for the Central Universities (Grant No. 40500-20104-222609).

\bibliographystyle{splncs04}
\bibliography{mybibliography}

\newpage

\appendix
\setcounter{figure}{5}
\setcounter{table}{2}
\setcounter{theorem}{3} 
\setcounter{lemma}{2}   

\section{Experiment Setup}

\subsection{\textbf{Experiment Setup of LIA}} \label{setup of LIA}

For the GI-LIA configuration, we set the maximum number of queries for the radius binary search to 100. The center of the second circle ($A_2$) is determined by generating four surrounding probes. The initial step size for these probes is equal to the radius, with a maximum of 10 retry rounds. If a round fails, the probe generation step size for the subsequent round is reduced by a factor of 0.8.

For the ZO-LIA configuration, which builds upon the GI-LIA settings, we set the number of rank optimization iterations to 10, with 4 probes generated per iteration. Furthermore, we filter out any probes that yield a rank worse than the best result from the previous iteration, thereby preventing them from contributing to the estimation of the current pseudo-gradient direction. The learning rate for this optimization process is set to 0.005.

\subsection{Experiment Setup of DPRS} \label{Experiment Setup of DPRS}

\subsubsection{\textbf{Baseline}}
 There are three baselines in our paper: \textit{1) 2D Laplace mechanism}\cite{andres2013geo}:  Perturbs true coordinates by adding noise from a two-dimensional Laplace distribution. \textit{2) SRR mechanism} \cite{wang2022srr}: Assigns a staircase of perturbation probabilities to location groups based on their distance from the true coordinates. \textit{3) Square mechanism} \cite{hong2022collecting}: Employs two distinct probabilities to sample locations, one for a square region around the user and another for the rest of the domain.

\subsubsection{\textbf{Utility Metric}} \label{metric}
For the k-NN task, we use the recall and distance ratio as the utility metrics to evaluate the performance. Given a query point $q$, the ground truth set is denoted as $G(q)$, and the set returned by the k-NN algorithm is denoted as $P(q)$. The recall and distance ratio are calculated as follows:
\begin{equation}
\label{eq:recall}
recall(q,k) = \frac{|P(q)\cap G(q)|}{k} .
\end{equation}
\begin{equation}
\label{eq:distance ratio}
distance\ ratio(q) = \frac{\sum_{g\in G(q)}{||g-q||_2}}{\sum_{p\in P(q)}{||p-q||_2}}. 
\end{equation}

\subsubsection{\textbf{Parameter Settings}}
 In our privacy budget design, we set $\epsilon_{p} = \epsilon_{c}$ and $\delta=10^{-5}$. The number of iterations for Algorithm \ref{algo:establishing centers} is set to 12. The radius scaling ranges from 0.2 to 0.8 times the initial radius, where the initial radius is defined as the distance from a given center to its nearest center. For location setting, we normalize the data to the range $[-1, 1]$.

\section{Gaussian Rejection Sampling Mechanism}\label{appendix:Gau}
Algorithm~\ref{alg:rejection_sampling_single_gaussian} describes the pipeline of applying DPRS to the Gaussian distribution. The main logic aligns with Algorithm \ref{alg:rejection_sampling_single}. The sensitivity of DPRS-Gau is the $L_2$ norm of the original range (i.e., $\sqrt{2}r$).

\section{Theoretical Analysis and Proofs}
\subsection{Theorem for Gaussian Rejection Sampling Mechanism} \label{Gaussian RSM}
\begin{algorithm}
    \caption{Gaussian-RSM}
    \label{alg:rejection_sampling_single_gaussian}
    
    \SetKwInOut{KwIn}{Input}
    \SetKwInOut{KwOut}{Output}
    
    \KwIn{
        data $\mathbf{x}$, interval $I$ (center $\mathbf{x}_c$, radius $R$), sensitivity $||\mu||_2 = \sqrt{2}r$, Gaussian distribution $p(\cdot|\mathbf{x},\sigma^2I_d)$, target function $f(\cdot)\propto p(\cdot|\mathbf{x},\sigma^2I_d)$.
    }
    \KwOut{
        perturbed data $\mathbf{x}'$.
    }
    
    \If{$\|\mathbf{x} - \mathbf{x}_c\|_k \le R$}{
    
        $M \leftarrow 1$.
        
    }
    \Else{
        $B \leftarrow \text{GetBoundary}(\mathbf{x}_c,R)$.
        
        $\mathbf{z}^* \leftarrow \min_{\mathbf{z} \in B} \|\mathbf{z} - \mathbf{x}\|_2$.
        
        $M \leftarrow f(\mathbf{z}^*)$.
    }
    
    \While{true}{
        $u,u_1,u_2\leftarrow \text{RandomUniform}(0, 1)$.
        
        $R \leftarrow R \sqrt{u_1}$, $\theta \leftarrow 2 \pi u_2$.

        $\mathbf{x}' \leftarrow \mathbf{x}_c + (R \cos(\theta),R \sin(\theta))$.
        
        \If{$u < \frac{f(\mathbf{x}')}{M}$}{
        
            \Return{$\mathbf{x}'$.}
            
        }
    }
    
\end{algorithm}

\begin{theorem} \label{definition:rejection_sampling_gaussian}
Let $p(\mathbf{t}|\mathbf{x},\sigma ^2)$ be a probability density function with support on a bounded domain $D \subset \mathbb{R}^2$, defined as:
\begin{equation}
\label{eq:q(x)}
p(\mathbf{t}|\mathbf{x},\sigma ^2 I_d) = 
\begin{cases} 
\frac{1}{N_C} Gau(\mathbf{t}|\mathbf{x},\sigma ^2 I_d) & \text{if } \mathbf{t} \in D \\
0 & \text{otherwise}, 
\end{cases}
\end{equation}
where $N_C = \iint_D Gau(\mathbf{t}|\mathbf{x},\sigma ^2I_d) \,d\mathbf{t}$. If a sample $\mathbf{x}'$ is generated by Algorithm \ref{alg:rejection_sampling_single_gaussian}, then the probability distribution of $\mathbf{x}'$ follows distribution $p(\mathbf{t}|\mathbf{x},\sigma ^2I_d)$.

The detailed proof of Theorem \ref{definition:rejection_sampling_gaussian} is provided in Theorem \ref{definition:rejection_sampling_single}.
\end{theorem}

\begin{theorem} \label{definition:selective publish of Gaussian}
Algorithm ~\ref{alg:rejection_sampling_single_gaussian} satisfies $(\alpha,{\alpha ||\mu||_2 ^2}/{2\sigma^2})-RDP$.

\begin{proof}
let $D_{\alpha}(S)$ be RDP of Algrithm \ref{alg:rejection_sampling_single_gaussian} according to Lemma \ref{theorem:rdp_new_feature}, we have: 
\begin{align}
\begin{split}
\begin{aligned}
&D_{\alpha}(S) \\
&\leq D_{\alpha}(N(0,\sigma^2I_d)||N(\mu,\sigma^2I_d))\\
&\leq \frac{1}{\alpha-1}\log\{\frac{1}{(2\pi\sigma^2)^{d/2}}\cdot \int_{\mathbb{R}^d}\exp\left(-\frac{\alpha\|x-\mu\|_2^2+(1-\alpha)\|x\|_2^2}{2\sigma^2}\right)d^dx\} \\
&\leq \frac{1}{\alpha-1}\log\{\frac{1}{(2\pi\sigma^2)^{d/2}} \cdot\int_{\mathbb{R}^d}\exp\left(-\frac{\alpha(x^Tx-2\mu^Tx+\mu^T\mu)+(1-\alpha)x^Tx}{2\sigma^2}\right)d^dx\}\\
&\leq \frac{1}{\alpha-1}\log\{\frac{1}{(2\pi\sigma^2)^{d/2}} \cdot\int_{\mathbb{R}^d}\exp\left(-\frac{x^Tx-2\alpha\mu^Tx+\alpha\|\mu\|^2}{2\sigma^2}\right)d^dx\} \\
&\leq \frac{1}{\alpha-1}\log\{\frac{1}{(2\pi\sigma^2)^{d/2}} \cdot\int_{\mathbb{R}^d}\exp\left(-\frac{||x-\alpha\mu||^2-\alpha(\alpha-1)||\mu||^2}{2\sigma^2}\right)d^dx\} \\
& \leq \frac{1}{\alpha-1}\log\{\exp({\frac{\alpha(\alpha-1)||\mu||^2}{2\alpha ^2}})\frac{1}{(2\pi\sigma^2)^{d/2}} \cdot \int_{\mathbb{R}^d}\exp\left(-\frac{||x-\alpha\mu||^2}{2\sigma^2}\right)d^dx\}\\
&\leq \frac{\alpha ||\mu||_2^2}{2\alpha ^2}.
\end{aligned}
\end{split}
\end{align}
So, Theorem \ref{definition:selective publish of Gaussian} is proved.
\end{proof}
\end{theorem}

\begin{theorem}\label{thm:utility-gau-truncated}
Let $\mathbf{x_c} =(0,0)$, $R=1$. For any input point $\mathbf{x} \in [-1, 1] \times [-1, 1]$, Algorithm~\ref{alg:rejection_sampling_single_gaussian} produces a perturbed output $\mathbf{x}' = A(\mathbf{x})$. Then the expected distance between $\mathbf{x}$ and $\mathbf{x}'$ is bounded as follows:
\begin{equation}
\label{eq:perturbation_bound_en}
\mathbb{E}_{A}[|x' - x|] \le \sigma\sqrt{\frac{2}{\pi}} \cdot \frac{1 - e^{-\frac{2}{\sigma^2}}}{\Phi\left(\frac{2}{\sigma}\right) - \frac{1}{2}}.
\end{equation}

\begin{proof}
In fact, no analytical solution exists for the rejection sampling mechanism on the unit disk. Therefore, we replace the perturbation interval with a square region $S=[-1,1]^2$. In fact, we have:   
\begin{align}
\begin{split}
\begin{aligned}
E_A[|x-x'|_1] & \leq \frac{\iint_S\|y-x\|_1\exp\left(-\frac{\|y-x\|_2^2}{2\sigma^2}\right)dy}{\iint_S\exp\left(-\frac{\|y-x\|_2^2}{2\sigma^2}\right)dy} = \frac{N_{2D}(x)}{C_{2D}(x)}. 
\end{aligned}
\end{split}
\end{align}

\begin{align}
\begin{split}
\begin{aligned}
C_{2D}(x)&=\int_{-1}^1\int_{-1}^1\exp(-\frac{(y_1-x_1)^2}{2\sigma^2})\exp(-\frac{(y_2-x_2)^2}{2\sigma^2})dy_1dy_2\\
&= \left(\int_{-1}^1\exp(-\frac{(t-x_1)^2}{2\sigma^2})dt\right)\cdot\left(\int_{-1}^1\exp(-\frac{(t-x_2)^2}{2\sigma^2})dt\right)\\
&=C_{1D}(x_1)\cdot C_{1D}(x_2).
\end{aligned}
\end{split}
\end{align}

\begin{align}
\begin{split}
\begin{aligned}
N_{2D}(x)&=\iint_S(|y_1-x_1|+|y_2-x_2|)\exp(-\frac{\|y-x\|_2^2}{2\sigma^2})dy \\
&=\iint_S|y_1-x_1|\exp(-\ldots)dy+\iint_S|y_2-x_2|\exp(-\ldots)dy\\
&=\left(\int_{-1}^1|t-x_1|e^{-\frac{(t-x_1)^2}{2\sigma^2}}dt\right)\cdot\left(\int_{-1}^1e^{-\frac{(t-x_2)^2}{2\sigma^2}}dt\right) \\
& + \left(\int_{-1}^1|t-x_2|e^{-\frac{(t-x_2)^2}{2\sigma^2}}dt\right)\cdot\left(\int_{-1}^1e^{-\frac{(t-x_1)^2}{2\sigma^2}}dt\right)\\
&=N_{1D}(x_1)C_{1D}(x_2)+N_{1D}(x_2)C_{1D}(x_1).
\end{aligned}
\end{split}
\end{align}

\begin{align}
\begin{split}
\begin{aligned}
E_A[|x-x'|_1] &\leq \frac{N_{1D}(x_1)C_{1D}(x_2)+C_{1D}(x_1)N_{1D}(x_2)}{C_{1D}(x_1)C_{1D}(x_2)} \\
&=\frac{N_{1D}(x_1)}{C_{1D}(x_1)}+\frac{N_{1D}(x_2)}{C_{1D}(x_2)}.
\end{aligned}
\end{split}
\end{align}

Let $E_{1D}(x) = \frac{N_{1D}(x)}{C_{1D}(x)}$, in fact we have:
\begin{align}
\begin{split}
\begin{aligned}
E_{1D}(x)\leq E_{1D}(1) &= \frac{\int_{-1}^1|y-1|\exp\left(-\frac{(y-1)^2}{2\sigma^2}\right)dy}{\int_{-1}^1\exp\left(-\frac{(t-1)^2}{2\sigma^2}\right)dt}\\
& = \frac{\int_{-2}^{0}|z|\exp\left(-\frac{z^{2}}{2\sigma^{2}}\right)dz}{\int_{-2}^{0}\exp\left(-\frac{z^{2}}{2\sigma^{2}}\right)dz}\\
& =\frac{\sigma^2\left(1-e^{-2/\sigma^2}\right)}{\sigma\int_{-2/\sigma}^0e^{-t^2/2}dt}\\
&=\frac{\sigma\left(1-e^{-2/\sigma^2}\right)}{\sqrt{2\pi}\left(\Phi(2/\sigma)-1/2\right)}.
\end{aligned}
\end{split}
\end{align}

Therefore, Theorem \ref{thm:utility-gau-truncated} is proved as follows:
\begin{align}
\begin{split}
\begin{aligned}
E_A[|x-x'|_1] \leq 2E_{1D}(x) \leq \sigma\sqrt{\frac{2}{\pi}} \cdot \frac{1 - e^{-\frac{2}{\sigma^2}}}{\Phi\left(\frac{2}{\sigma}\right) - \frac{1}{2}}.
\end{aligned}
\end{split}
\end{align}

\end{proof}

\end{theorem}

\subsection{Lemma and Theorem for Laplace Rejection Sampling Mechanism} \label{subsec:Lemma theorem supplementation}

\begin{lemma} \label{theorem_A}
For any interval $S$ and $\alpha >1$, we have:
\begin{align}
\begin{split}
\begin{aligned}
I_\alpha(S)\cdot Z_Q(S)^{\alpha-1}\geq Z_P(S)^\alpha,
\end{aligned}
\end{split}
\end{align}
where $I_\alpha(S)=\int_Sp(x)^\alpha q(x)^{1-\alpha}dx$, $Z_P(S)=\int_Sp(x)dx$, $Z_Q(S)=\int_Sq(x)dx$.

\begin{proof}
we define $g(x)=\frac{q(x)}{Z_Q(S)}$, $h(x)=\frac{p(x)}{q(x)}$ and a convex function $f(t) = t^\alpha$. 

\begin{definition}[\textbf{Jensen's Inequality}] \label{definition_A}
If $f(\cdot)$ is a convex function, then for any $g(x)$, $h(x)$, we have:
\begin{align}
\begin{split}
\begin{aligned}
f(\int h(x)g(x)dx)\leq \int f(h(x))g(x)dx.
\end{aligned}
\end{split}
\end{align}
\end{definition}

According to Definition \ref{definition_A} above, we have:
\begin{align}
\begin{split}
\begin{aligned}
&f(\int_Sh(x)g(x)dx)\leq\int_Sf(h(x))g(x)dx \\
&\iff f\left(\int_S\frac{p(x)}{q(x)}\frac{q(x)}{Z_Q(S)}dx\right) \leq \int_S\left(\frac{p(x)}{q(x)}\right)^\alpha\frac{q(x)}{Z_Q(S)}dx \\
& \iff f\left(\frac{1}{Z_Q(S)}\int_Sp(x)dx\right) \leq\frac{1}{Z_Q(S)}\int_Sp(x)^\alpha q(x)^{1-\alpha}dx \\
&\iff \left(\frac{Z_P(S)}{Z_Q(S)}\right)^\alpha \leq\frac{I_\alpha(S)}{Z_Q(S)} \\
& \iff I_\alpha(S)\cdot Z_Q(S)^{\alpha-1}\geq Z_P(S)^\alpha.
\end{aligned}
\end{split}
\end{align}
So Lemma \ref{theorem_A} is proved.
\end{proof}

\end{lemma}

\begin{lemma} \label{theorem_B}
for any $x,y>0$ and $\alpha >1$, we have:
\begin{align}
\begin{split}
\begin{aligned}
(1+x)^\alpha(1+y)^{1-\alpha}-x^\alpha y^{1-\alpha}\leq 1.
\end{aligned}
\end{split}
\end{align}

\begin{proof}
we define $g(x)=(1+x)^\alpha(1+y)^{1-\alpha}-x^\alpha y^{1-\alpha}$, we have:
\begin{align}
\begin{split}
\begin{aligned}
&g^{\prime}(x)=\alpha(1+x)^{\alpha-1}(1+y)^{1-\alpha}-\alpha x^{\alpha-1}y^{1-\alpha}\\
&\textbf{let}\quad g^{\prime}(x) = 0 \\
& \iff \alpha(1+x)^{\alpha-1}(1+y)^{1-\alpha}=\alpha x^{\alpha-1}y^{1-\alpha} \\
& \iff \left(\frac{1+x}{x}\right)^{\alpha-1}=\left(\frac{1+y}{y}\right)^{\alpha-1} \\
& \iff x=y.
\end{aligned}
\end{split}
\end{align}
When $x>y$, $g'(x)<0$; when $x<y$, $g'(x)>0$
So we have:
\begin{align}
\begin{split}
\begin{aligned}
g(x) \leq g(y) &= (1+y)^\alpha(1+y)^{1-\alpha}-y^\alpha y^{1-\alpha}\\
&=(1+y)^1-y^1\\
&=1.
\end{aligned}
\end{split}
\end{align}
so, Lemma \ref{theorem_B} is proved.
\end{proof}
\end{lemma}

\begin{lemma} \label{theorem:rdp_new_feature}
let $D_\alpha(S)=\frac{1}{\alpha-1}\log\left(\frac{I_\alpha(S)}{Z_P(S)^\alpha Z_Q(S)^{1-\alpha}}\right)$, if $S_1 \subseteq S_2$, then $D_\alpha(S_1)\leq D_\alpha(S_2)$. where $I_\alpha(S)=\int_Sp(x)^\alpha q(x)^{1-\alpha}dx$, $Z_P(S)=\int_Sp(x)dx$, $Z_Q(S)=\int_Sq(x)dx$.

\begin{proof}

let $S_2 = S_1 \cup S_{rem}$, we can prove this as follows:
\begin{align}
\begin{split}
\begin{aligned}    
&D_\alpha(S_1)\leq D_\alpha(S_2) \\
&\iff \frac{I_{\alpha}(S_1)}{Z_P(S_1)^\alpha Z_Q(S_1)^ {1-\alpha}} \leq \frac{I_{\alpha}(S_2)}{Z_P(S_2)^\alpha Z_Q(S_2)^ {1-\alpha}} \\
& \iff (\frac{Z_P(S_2)}{Z_P(S_1)})^\alpha(\frac{Z_Q(S_2)}{Z_Q(S_1)})^{1-\alpha} \leq \frac{I_{\alpha}(S_2)}{I_{\alpha}(S_1)} \\
& \iff (1+\frac{Z_P(S_{rem})}{Z_P(S_1)})^\alpha(1+\frac{Z_Q(S_{rem})}{Z_Q(S_1)})^{1-\alpha} \leq 1 + \frac{I_{\alpha}(S_{rem})}{I_{\alpha}(S_1)}.
\end{aligned}
\end{split}
\end{align}  

According to Lemma \ref{theorem_A} we have:
\begin{align}
\begin{split}
\begin{aligned}
&I_\alpha(S_{rem})\cdot Z_Q(S_{rem})^{\alpha-1}\geq Z_P(S_{rem})^\alpha \\
&I_\alpha(S_1)\cdot Z_Q(S_1)^{\alpha-1}\geq Z_P(S_1)^\alpha \\
& \iff \frac{I_{\alpha}(S_{rem})}{I_{\alpha}(S_1)} \geq (\frac{Z_P(S_{rem})}{Z_P(S_1)})^\alpha \cdot (\frac{Z_Q(S_{rem})}{Z_Q(S_1)})^{1-\alpha}.
\end{aligned}
\end{split}
\end{align}

So, let $x =\frac{Z_P(S_{rem})}{Z_P(S_1)}$, $y =\frac{Z_Q(S_{rem})}{Z_Q(S_1)}$, we need to prove
\begin{align}
\begin{split}
\begin{aligned}
1+ x^\alpha y^{1-\alpha} \geq (1+x)^\alpha(1+y)^{1-\alpha}.
\end{aligned}
\end{split}
\end{align}
According to Lemma \ref{theorem_B}, Lemma \ref{theorem:rdp_new_feature} is proved.
\end{proof}

\end{lemma}

\begin{theorem} \label{definition:rejection_sampling_single}
Let $p(\mathbf{t}|\mathbf{x},\lambda)$ be a probability density function with support on a bounded domain $D \subset \mathbb{R}^2$, defined as:
\begin{equation}
\label{eq:p(x)}
p(\mathbf{t}|\mathbf{x},\lambda) = 
\begin{cases} 
\frac{1}{N_C} Lap(\mathbf{t}|\mathbf{x},\lambda) & \text{if } \mathbf{t} \in D \\
0 & \text{otherwise},
\end{cases}
\end{equation}
where $N_C = \iint_D Lap(\mathbf{t}|\mathbf{x},\lambda) \,d\mathbf{t}$.

If a sample $\mathbf{x}'$ is generated by Algorithm \ref{alg:rejection_sampling_single}, then the probability distribution of $\mathbf{x}'$ follows target distribution $p(\mathbf{t}|\mathbf{x},\lambda)$.

\begin{proof}
To prove Theorems \ref{definition:rejection_sampling_gaussian} and Theorems \ref{definition:rejection_sampling_single}, we need to establish a more general result from which they follow. let noisy distribution be $p(x)$ as follows:
\begin{equation}
\label{eq:p_gen(x)}
p(\mathbf{x}) = 
\begin{cases} 
\frac{1}{N_C} f^*(\mathbf{x}) & \text{if } \mathbf{x} \in D \\
0 & \text{otherwise} ,
\end{cases}
\end{equation}
where $N_C = \int_D f^*(\mathbf{z}) \,d\mathbf{z}$ is the normalization constant.

we assume $D$ is a Euclidean
sphere with a radius $R$ and $Y$ is the output of RSM. We will prove that $Pr[Y=x|accept] = p(x)$, in fact, we hava:
\begin{align}
\begin{split}
\begin{aligned}
Pr[Y=x|accept] &= \frac{Pr[accept|Y=x]\cdot Pr[Y=x]}{Pr[accpet]} \\
&=\frac{Pr[accept|Y=x]\cdot \frac{1}{\pi R^2}}{Pr[accpet]}.
\end{aligned}
\end{split}
\end{align}
For the first term in the numerator, we can calculate it in the following way:
\begin{align}
\begin{split}
\begin{aligned}
Pr[accept|Y=x]  = Pr[u<\frac{f^*(x)}{M}] = \frac{f^*(x)}{M},
\end{aligned}
\end{split}
\end{align}
where $M$ is the max probability value in $D$, and $u \sim Uniform(0,1)$.

For the denominator, we can calculate it in the following manner:
\begin{align}
\begin{split}
\begin{aligned}
Pr[accept] &= \int_{x\in D} {Pr[accept|Y=x] \cdot Pr[Y=x]}dx \\
&=\int_{x \in D} \frac{f^*(x)}{M}\cdot \frac{1}{\pi R^2} dx \\
&=\frac{\int_{x\in D}{f^*(x)}dx}{M\pi R^2} \\
&=\frac{N_c}{M\pi R^2}.
\end{aligned}
\end{split}
\end{align}
So, we have:
\begin{align}
\begin{split}
\begin{aligned}
Pr[Y=x|accept]  &= \frac{Pr[accept|Y=x]\cdot \frac{1}{\pi R^2}}{Pr[accpet]} \\
&=\frac{f^*(x)}{M} \cdot \frac{1}{\pi R^2} \cdot \frac{M\pi R^2}{N_C}\\
&=\frac{f^*(x)}{N_c} \\
&=p(x).
\end{aligned}
\end{split}
\end{align}
Therefore, Theorem \ref{definition:rejection_sampling_gaussian} and Theorem \ref{definition:rejection_sampling_single} are proved.
\end{proof}

\end{theorem}

\begin{theorem}\label{thm:utility-lap-truncated-2d}
Let $\mathbf{x_c} =(0,0)$, $R=1$. For any input point $\mathbf{x} \in [-1, 1] \times [-1, 1]$, Algorithm \ref{alg:rejection_sampling_single} produces a perturbed output $\mathbf{x}' = A(\mathbf{x})$. Then the expected distance between $\mathbf{x}$ and $\mathbf{x}'$ is bounded as follows:
\begin{equation}
\label{eq:perturbation_bound_en_2d}
\mathbb{E}_{A}[\|\mathbf{x}' - \mathbf{x}\|_1] \le 2(\lambda - \frac{2}{e^{2/\lambda} - 1}).
\end{equation}

\begin{proof}
 In fact, the rejection sampling mechanism on the unit circle has no analytical solution. For this reason, we choose to calculate the expectation of the analytical expression on a square region $S=[-1,1]^2$, because the region S is slightly larger than the unit circle, from which an upper bound can be obtained:
\begin{align}
\begin{split}
\begin{aligned}
E_A[|x-x'|_1] \leq \frac{\iint_S\|y-x\|_1\exp(-\frac{\|y-x\|_1}{\lambda})dy}{\iint_S\exp(-\frac{\|y-x\|_1}{\lambda})dy}.
\end{aligned}
\end{split}
\end{align}

Let $C_{2D}(x) = \iint_S\|y-x\|_1\exp(-\frac{\|y-x\|_1}{\lambda})dy$, we have:
\begin{align}
\begin{split}
\begin{aligned}
C_{2D}(x)&=\int_{-1}^1\int_{-1}^1\exp\left(-\frac{|y_1-x_1|}{\lambda}\right)\exp\left(-\frac{|y_2-x_2|}{\lambda}\right)dy_1dy_2\\
&=\left(\int_{-1}^1\exp\left(-\frac{|t-x_1|}{\lambda}\right)dt\right)\cdot\left(\int_{-1}^1\exp\left(-\frac{|t-x_2|}{\lambda}\right)dt\right)\\
&=C_{1D}(x_1)\cdot C_{1D}(x_2).
\end{aligned}
\end{split}
\end{align}
\quad Let $N_{2D}(x) = \iint_S\exp(-\frac{\|y-x\|_1}{\lambda})dy$, we have:
\begin{align}
\begin{split}
\begin{aligned}
N_{2D}(x)&=\iint_S(\sum_i{|y_i-x_i|})\Pi_i{\exp\left(-\frac{|y_i-x_i|}{\lambda}\right)}dy_1dy_2\\
&=\left(\int_{-1}^1|y_1-x_1|e^{-\frac{|y_1-x_1|}{\lambda}}dy_1\right)\cdot\left(\int_{-1}^1e^{-\frac{|y_2-x_2|}{\lambda}}dy_2\right)\\ 
&+\left(\int_{-1}^1e^{-\frac{|y_2-x_1|}{\lambda}}dy_2\right) \cdot \left(\int_{-1}^1|y_1-x_2|e^{-\frac{|y_1-x_2|}{\lambda}}dy_1\right)\\
&=N_{1D}(x_1)C_{1D}(x_2)+C_{1D}(x_1)N_{1D}(x_2).
\end{aligned}
\end{split}
\end{align}
Hence, the original formula is represented as follows:
\begin{align}
\begin{split}
\begin{aligned}
E_A[|x-x'|_1] &\leq \frac{N_{1D}(x_1)C_{1D}(x_2)+C_{1D}(x_1)N_{1D}(x_2)}{C_{1D}(x_1)C_{1D}(x_2)}\\
&=\frac{N_{1D}(x_1)}{C_{1D}(x_1)}+\frac{N_{1D}(x_2)}{C_{1D}(x_2)}.
\end{aligned}
\end{split}
\end{align}
Let $E_{1D}(x) = \frac{N_{1D}(x)}{C_{1D}(x)}$, in fact we have:
\begin{align}
\begin{split}
\begin{aligned}
E_{1D}(x)\leq E_{1D}(1) &=\frac{\int_{-1}^1|y-1|\exp\left(-\frac{|y-1|}{\lambda}\right)dy}{\int_{-1}^1\exp\left(-\frac{|t-1|}{\lambda}\right)dt}\\
&=\frac{\int_2^0ze^{-z/\lambda}(-dz)}{\int_2^0e^{-z/\lambda}(-dz)}\\
&=\frac{\lambda^2-\lambda(\lambda+2)e^{-2/\lambda}}{\lambda(1-e^{-2/\lambda})}\\
&=\lambda-\frac{2}{e^{2/\lambda}-1}.
\end{aligned}
\end{split}
\end{align}
Therefore, Theorem \ref{thm:utility-lap-truncated-2d} is proved as follows:
\begin{align}
\begin{split}
\begin{aligned}
E_A[|x-x'|_1] \leq 2E_{1D}(x) \leq2(\lambda - \frac{2}{e^{2/\lambda} - 1}).
\end{aligned}
\end{split}
\end{align}   
\end{proof}
\end{theorem}

\section{Additional Experiment Results} \label{additional experiment}

\subsubsection{\textbf{Experiments of RSM-Gau.}}
Figures \ref{fig:knn_cluster_gaussian} and \ref{fig:RSM-Gau-k_variation} demonstrate the performance of integrating RSM-Gau into DPRS. Specifically, the former shows that the recall of DPRS-Gau increases as the number of centers grows from 600 to 1000 and larger privacy budgets ($\epsilon$) also yield higher recall, while the latter indicates that both DPRS-Gau and DPRS-Lap outperform the baselines (SRR, Square, and Laplace) across all datasets, confirming the effectiveness of the DPRS framework. 

\begin{figure*}[!t]
    \includegraphics[width=0.6\textwidth]{figure/epsilon_legend.pdf}
  \centering

  \begin{subfigure}{0.26\linewidth}
    \centering
    \includegraphics[width=\linewidth]{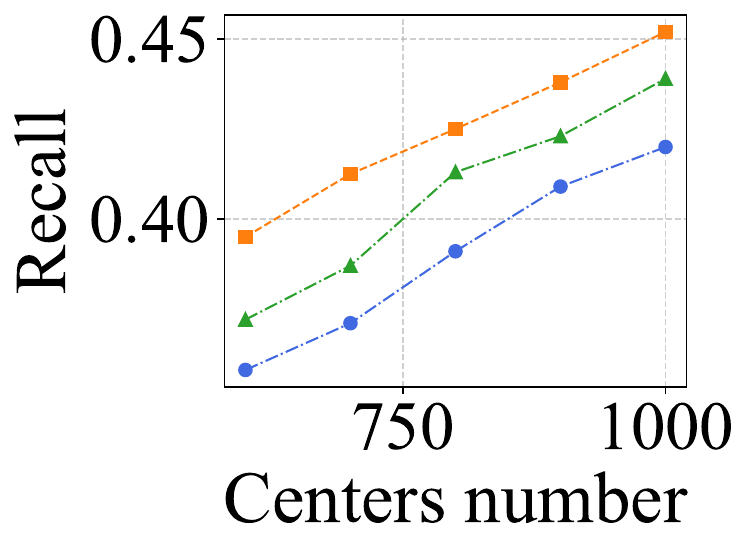}
    \captionsetup{skip=-2pt}
    \caption{Brightkite}
    \label{fig:span_gau_sub1}
  \end{subfigure}
  \begin{subfigure}{0.235\linewidth}
    \centering
    \includegraphics[width=\linewidth]{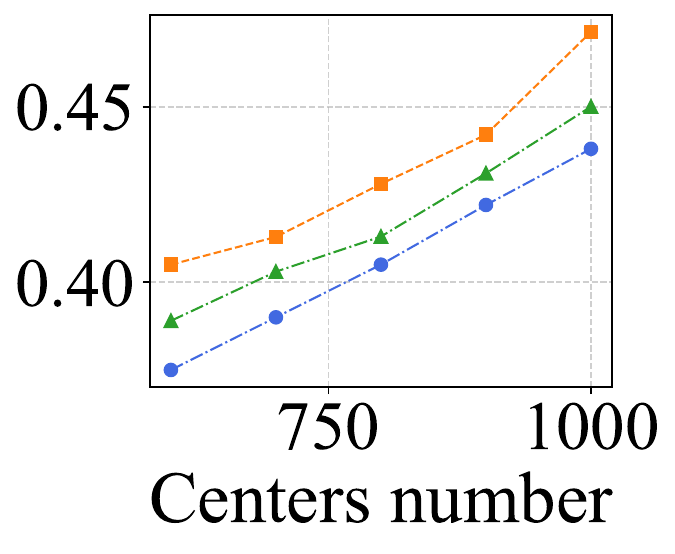}
    \captionsetup{skip=-2pt}
    \caption{Gowalla}
    \label{fig:span_gau_sub2}
  \end{subfigure}
  \begin{subfigure}{0.235\linewidth}
    \centering
    \includegraphics[width=\linewidth]{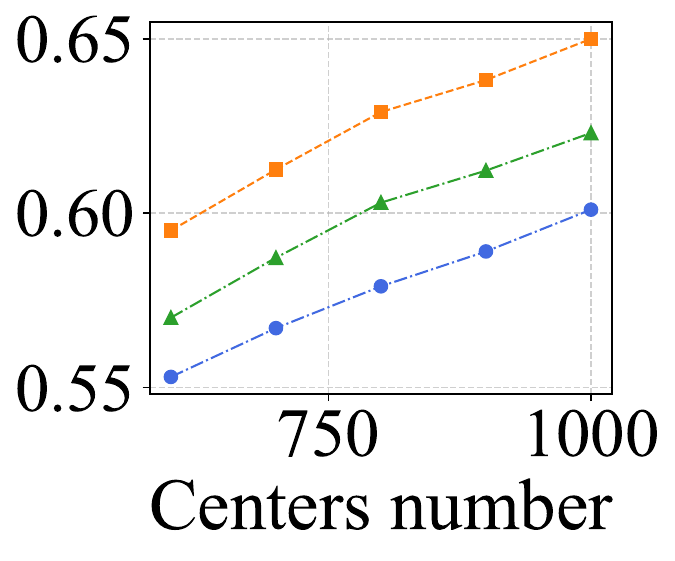}
    \captionsetup{skip=-2pt}
    \caption{Gaussian}
    \label{fig:span_gau_sub3}
  \end{subfigure}
  \begin{subfigure}{0.235\linewidth}
    \centering
    \includegraphics[width=\linewidth]{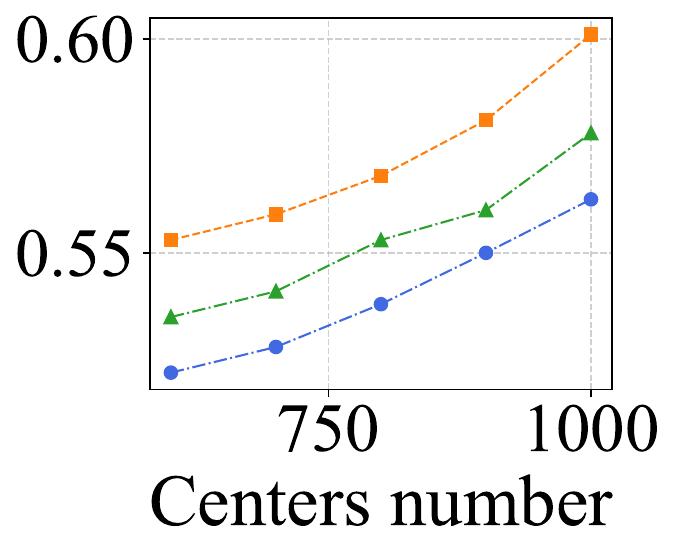}
    \captionsetup{skip=-2pt}
    \caption{Beta}
    \label{fig:span_gau_sub4}
  \end{subfigure}

  \caption{Performance impact on the number of interval centers for DPRS-Gau in 10-NN task at $k=10$.}
  \label{fig:knn_cluster_gaussian}
\end{figure*}

\begin{figure*}[!t]
    \includegraphics[width=0.85\textwidth]{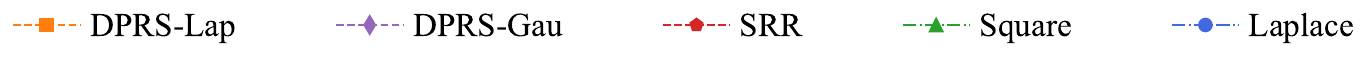}
  \centering

  \begin{subfigure}{0.26\linewidth}
    \centering
    \includegraphics[width=\linewidth]{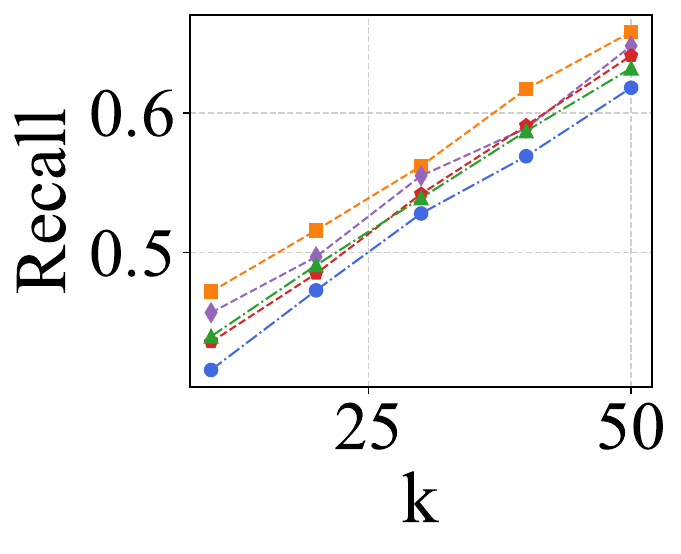}
    \captionsetup{skip=-2pt}
    \caption{Brightkite}
    \label{fig:k_1}
  \end{subfigure}
  \begin{subfigure}{0.235\linewidth}
    \centering
    \includegraphics[width=\linewidth]{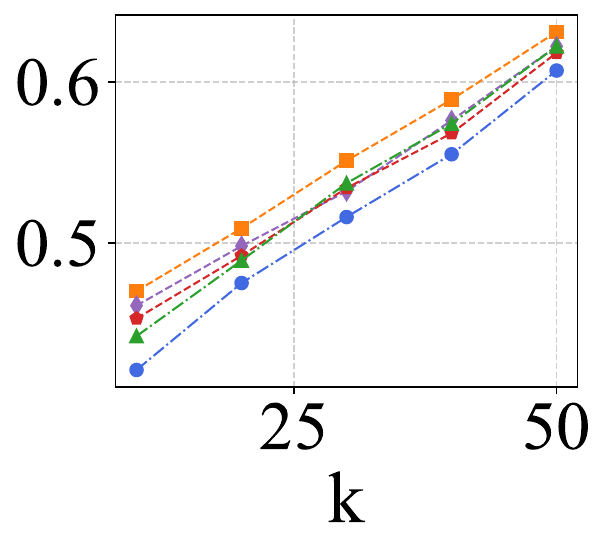}
    \captionsetup{skip=-2pt}
    \caption{Gowalla}
    \label{fig:span_sub2_gau_k}
  \end{subfigure}
  \begin{subfigure}{0.235\linewidth}
    \centering
    \includegraphics[width=\linewidth]{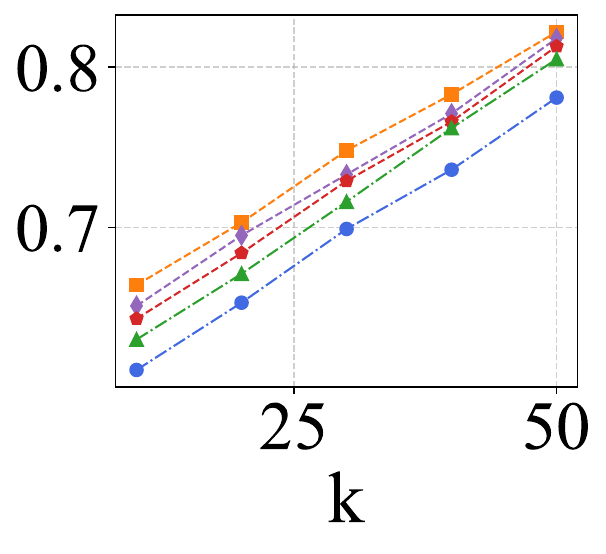}
    \captionsetup{skip=-2pt}
    \caption{Gaussian}
    \label{fig:span_sub3_gau_k}
  \end{subfigure}
  \begin{subfigure}{0.235\linewidth}
    \centering
    \includegraphics[width=\linewidth]{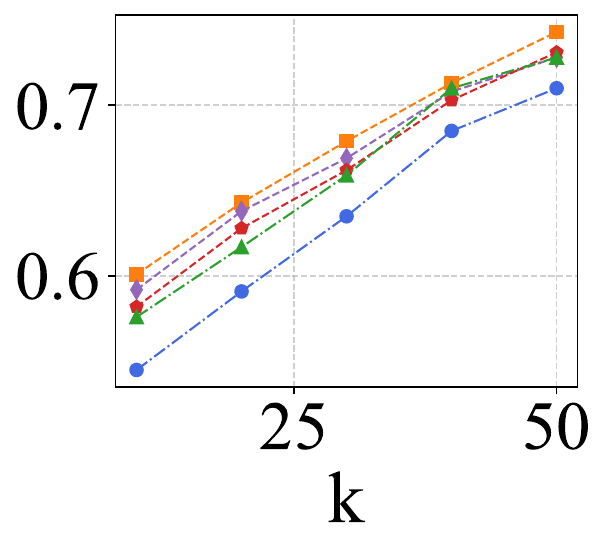}
    \captionsetup{skip=-2pt}
    \caption{Beta}
    \label{fig:span_sub4_gau_k}
  \end{subfigure}

  \caption{Performance impact on the number of neighbors for all mechanisms in the 10-NN task at $\epsilon=3$.}
  \label{fig:RSM-Gau-k_variation}
\end{figure*}

\begin{figure*}[!t] 
  \centering

  \begin{minipage}[t]{0.49\linewidth}
    \centering
    \includegraphics[width=0.85\linewidth]{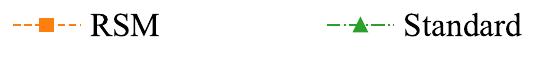}\vspace{1mm}

    \begin{subfigure}[t]{0.48\linewidth}
      \centering
      \includegraphics[width=\linewidth]{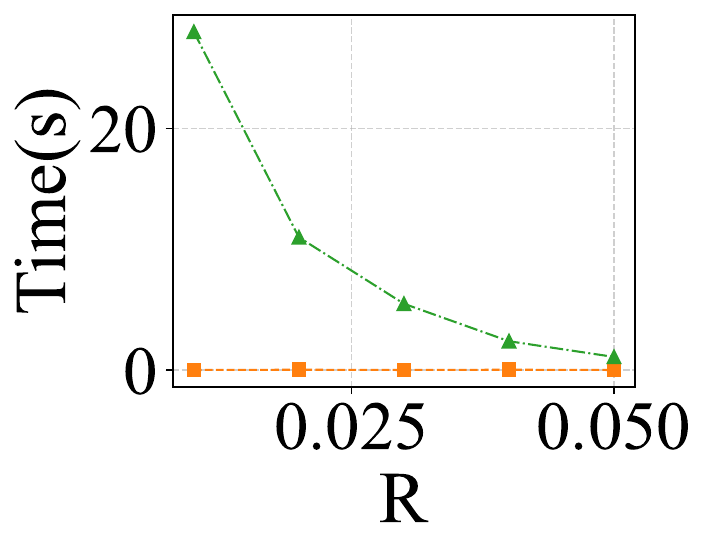}
      \caption{$\epsilon = 1$}
      \label{figure: time_eps=1}
    \end{subfigure}
    \hfill
    \begin{subfigure}[t]{0.48\linewidth}
      \centering
      \includegraphics[width=\linewidth]{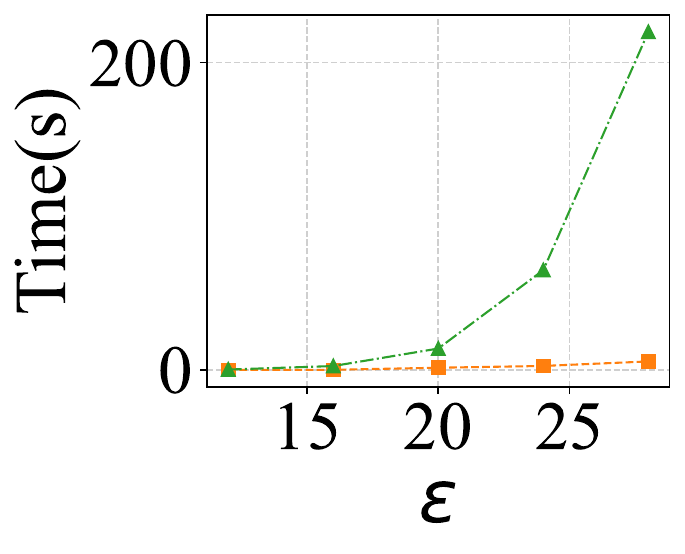}
      \caption{$R=0.5$}
      \label{figure: time_r=0.5}
    \end{subfigure}

    \caption{Time overhead for sampling algorithms.}
    \label{figure: time overhead}
  \end{minipage}
  \hfill
  \begin{minipage}[t]{0.49\linewidth}
    \centering
    \vspace{0pt}

    \begin{subfigure}[t]{0.48\linewidth}
      \centering
      \includegraphics[width=\linewidth]{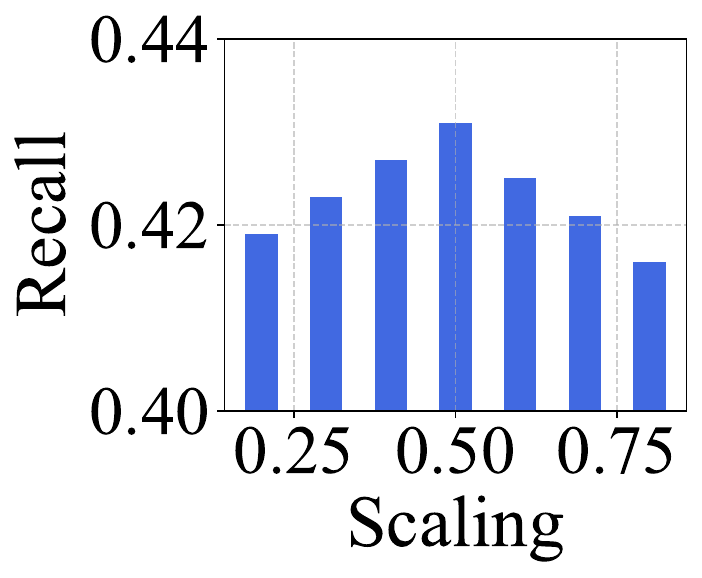}
      \caption{Recall}
      \label{figure: ratio_recall}
    \end{subfigure}
    \hfill
    \begin{subfigure}[t]{0.48\linewidth}
      \centering
      \includegraphics[width=\linewidth]{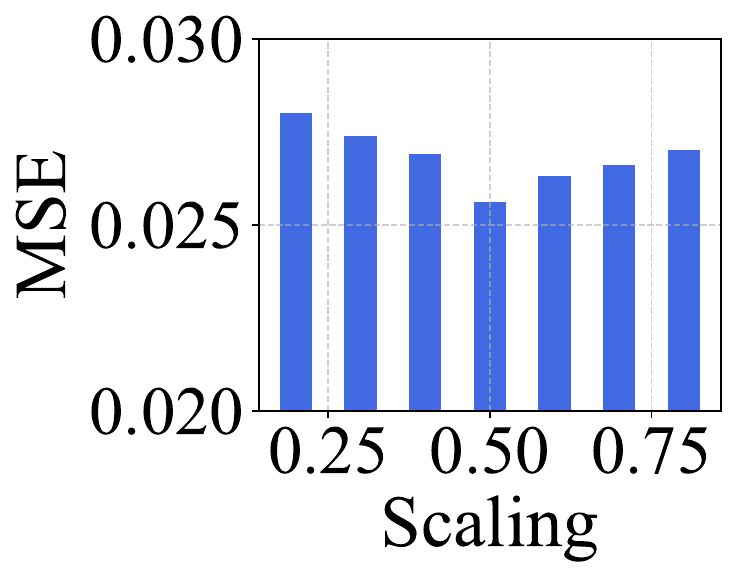}
      \caption{MSE}
      \label{figure: ratio_mse}
    \end{subfigure}

    \caption{Performance impact on radius scaling selection.}
    \label{figure: threshold selection}
  \end{minipage}

\end{figure*}

\subsubsection{\textbf{Time overhead analysis of DPRS.}}
We set the interval center to $(0,0)$ and uniformly sample 100 samples from the square $[-1, 1] \times [-1, 1]$, then compare the time cost of perturbing these samples into the interval. As shown in Figure \ref{figure: time overhead}, compared to the standard noise-and-resampling mechanism, our RSM consistently maintains a minimal time overhead (less than 40ms) under varying radius ($R$) and privacy budget ($\epsilon$) settings.

\subsubsection{\textbf{Optimal $\tau$ Analysis.}}
We evaluate the perturbation radius scaling at $\epsilon=1$. Figure \ref{figure: threshold selection} illustrates the impact of varying the scaling coefficient on $k$-NN recall and distance MSE. Specifically, setting the radius to 0.5 times the distance to the nearest PIC center yields the optimal performance, achieving the highest recall and the lowest MSE.

\end{document}